\documentclass[11pt]{article}
\usepackage[margin=1in]{geometry}

\usepackage{enumerate}
\usepackage[utf8]{inputenc}
\usepackage{amsmath, amsthm, amssymb,tabu, thm-restate}
\usepackage{color}
\usepackage{calc}
\usepackage{tikz}
\usetikzlibrary{positioning,arrows,decorations.pathreplacing,shapes}
\usetikzlibrary{calc}
\usepackage{multirow}
\usepackage{boxedminipage}
\usepackage{xifthen}
\usepackage{tabularx}
\usepackage{xcolor}
\usepackage{setspace}
\usepackage{authblk}
\usepackage{bm}
\usepackage{nameref}
\usepackage{framed}
\usepackage{mathpazo}
\usepackage{amsthm}
\usepackage{caption, subcaption}
\usepackage{framed}
\usepackage{algpseudocode}
\usepackage[algo2e, ruled, linesnumbered]{algorithm2e}

\usepackage[linesnumbered,ruled]{algorithm2e}
\usepackage{hyperref}
\usepackage{cleveref}

\newcommand{\ignore}[1]{}
\newlength{\bibitemsep}
\newlength{\bibparskip}
\let\oldthebibliography\thebibliography
\renewcommand\thebibliography[1]{%
	\oldthebibliography{#1}%
	\setlength{\parskip}{\bibitemsep}%
	\setlength{\itemsep}{\bibparskip}%
}

\newcommand{\R}{\mathbb{R}}

\newcommand{\x}{\mathbf{x}}
\newcommand{\y}{\mathbf{y}}

\newcommand{\pr}{\mathrm{par}}
\newcommand{\nsw}{\mathrm{NSW}}
\newcommand{\opt}{\mathrm{OPT}}
\newcommand{\cvx}{\mathrm{cvx}}
\newcommand{\ncvx}{\mathrm{ncvx}}

\newcommand{\RR}{\mathbb{R}}

\newcommand{\cA}{\mathcal{A}}

\newcommand{\cG}{\mathcal{G}}

\newcommand{\cP}{\mathcal{P}}
\newcommand{\cQ}{\mathcal{Q}}

\newcommand{\cS}{\mathcal{S}}

\newcommand{\kl}[2]{D_{\mathrm{KL}}(#1 \, || \,#2)}

\newcommand\mb[1]{\mathbf{#1}}

\newtheorem{claim}{Claim}[section]
\newtheorem{assumption}{Assumption}[section]

\newtheorem{fact}{Fact}[section]

\newtheorem{theorem}{Theorem}
\newtheorem{corollary}[theorem]{Corollary}

\newtheorem{definition}{Definition}
\newtheorem{lemma}[theorem]{Lemma}

\usepackage[textsize=tiny,textwidth=2cm,color=green!50!gray]{todonotes} %,disable

\newcommand*\samethanks[1][\value{footnote}]{\footnotemark[#1]}
\title{Approximation Algorithms for the Weighted Nash Social Welfare via Convex and Non-Convex Programs}

\author[1]{Adam Brown\thanks{ajmbrown@gatech.edu, msingh94@gatech.edu; supported in part by NSF CCF-2106444 and NSF CCF-1910423.}}
\author[2]{Aditi Laddha \thanks{aditi.laddha@yale.edu; supported in part by the Institute for Foundations of Data Science at Yale and  NSF CCF-2007443.}}
\author[3]{Madhusudhan Reddy Pittu \thanks{mpittu@andrew.cmu.edu; supported in part by NSF awards CCF-1955785 and CCF-2006953.}}
\author[1]{Mohit Singh\samethanks[1]}
\affil[1]{Georgia Institute of Technology}
\affil[2]{Yale University}
\affil[2]{Carnegie Mellon University}

\date{}

\begin{document}
\maketitle

\begin{abstract}
    In an instance of the weighted Nash Social Welfare problem, we are given a set of $m$ indivisible items, $\mathcal{G}$, and $n$ agents, $\mathcal{A}$, where each agent $i \in \mathcal{A}$ has a valuation $v_{ij}\geq 0$ for each item $j\in \mathcal{G}$. In addition, every agent $i$ has a non-negative weight $w_i$ such that the weights collectively sum up to $1$. The goal is to find an assignment $\sigma:\mathcal{G}\rightarrow \mathcal{A}$ that maximizes  $\prod_{i\in \mathcal{A}} \left(\sum_{j\in \sigma^{-1}(i)} v_{ij}\right)^{w_i}$, the product of the weighted valuations of the players. When all the weights equal $\frac1n$, the problem reduces to the classical Nash Social Welfare problem, which has recently received much attention. In this work, we present a $5\cdot\exp\left(2\cdot D_{\mathrm{KL}}(\mb{w}\, ||\, \frac{\vec{\mathbf{1}}}{n})\right) = 5\cdot\exp\left(2\log{n} + 2\sum_{i=1}^n w_i \log{w_i}\right)$-approximation algorithm for the weighted Nash Social Welfare problem, where $D_{\mathrm{KL}}(\mb{w}\, ||\, \frac{\vec{\mathbf{1}}}{n})$ denotes the KL-divergence between the distribution induced by $\mb{w}$ and the uniform distribution on $[n]$. 

    We show a novel connection between the convex programming relaxations for the unweighted variant of Nash Social Welfare presented in \cite{cole2017convex, anari2017nash}, and generalize the programs to two different mathematical programs for the weighted case. The first program is convex and is necessary for computational efficiency, while the second program is a non-convex relaxation that can be rounded efficiently. The approximation factor derives from the difference in the objective values of the convex and non-convex relaxation.
\end{abstract}

\section{Introduction}

In an instance of the weighted Nash Social Welfare problem, we are given a set of $m$ indivisible items $\cG$, and a set of $n$ agents, $\cA$. Every agent $i \in \cA$ has a weight $w_i \geq 0$ such that $\sum_{i \in \cA} w_i = 1$ and an additive valuation function $\mb{v}_i: 2^{\cG}\rightarrow \RR_{\geq 0}$. Let $v_{ij}:=\mb{v}_{i}(\{j\})$. The goal is to find an assignment of items, $\sigma: \cG \rightarrow \cA$, to maximize the following welfare function:
\begin{equation}
    \prod_{i\in \cA} \left( \sum_{j \in \sigma^{-1}(i)} v_{ij} \right)^{w_i} .
\end{equation}
For ease of notation, we will work with the log objective and denote 
\begin{equation}
    \nsw(\sigma)=\sum_{i\in \cA} w_i \, \log\left( \sum_{j \in \sigma^{-1}(i)} v_{ij} \right) .
\end{equation}
Let $\opt=\max_{\sigma: \cG \rightarrow \cA} \nsw(\sigma)$ denote the optimal log objective.
The case where $w_i=\frac1n$ for each $i\in \cA$ is the much-studied ``symmetric'' or unweighted Nash social welfare problem, where the objective is the geometric mean of agents' valuations. 

Fair and efficient division of resources among agents is a fundamental problem arising in various fields \cite{barbanel_taylor_2005, brams_taylor_1996, brandt_2016, robertson_webb_1998,rothe_2015, young_1994}.
While there are many social welfare functions which can be used to evaluate the efficacy of an assignment of goods to the agents, the Nash Social Welfare function is well-known to interpolate between fairness and overall utility. The unweighted Nash Social Welfare function first appeared as the solution to an arbitration scheme proposed by Nash for two-person bargaining games and was later generalized to multiple players~\cite{nash1950bargaining, kaneko1979nash}. Since then, it has been widely used in numerous fields to model resource allocation problems. An attractive feature of the objective is that it is invariant under scaling by any of the agent's valuations, and therefore, each agent can specify its utility in its own units (see~\cite{chae_moulin_2004} for a detailed treatment). While the theory of Nash Social Welfare objective was initially developed for divisible items, more recently, it has been applied in the context of indivisible items. We refer the reader to ~\cite{caragiannis2019unreasonable} for a comprehensive overview of the problem in the latter setting. Indeed, optimizing the Nash Social Welfare objective also implies notions of fairness, such as \emph{envy-free} allocation in an approximate sense~\cite{caragiannis2019unreasonable,barman2018finding}.
 
The Nash Social Welfare function with weights (also referred to as asymmetric or non-symmetric Nash Social Welfare) was first studied in the seventies \cite{harsanyi_selten_1972, kalai1977nonsymmetric} in the context of two-person bargaining games. For example, in the bargaining context, it allows different agents to have different weights. Due to this flexibility, problems in many diverse domains can be modeled using the weighted objective, including bargaining theory \cite{chae_moulin_2004,laruelle_valenciano_2007}, water resource allocation \cite{fu_leinberg_lavi_2012,houba_laan_zeng_2013}, and climate agreements \cite{yu_ierland_2017}. From a context of indivisible goods, the study of this problem has been much more recent~\cite{garg2020approximating,garg2021approximating,garg2023approximating}. In this work, we aim to shed light on the weighted Nash Social Welfare problem, mainly focusing on mathematical programming relaxations for the problem.

\subsection{Our Results and Contributions}

 Our main result is an
 $\exp\left(2\log{2} + \frac{1}{2e}+2\kl{\mb{w}}{\mb{u}}\right) \approx 4.81\cdot\exp\left(2\log{n} - 2\sum_{i=1}^n w_i \log {\frac{1}{{w_i}}}\right)$-approximation algorithm for the weighted Nash Social Welfare problem with additive valuations. When all the weights are the same, this gives a constant factor approximation. Our algorithm builds on and extends a convex programming relaxation for the unweighted variant of Nash Social Welfare presented in \cite{cole2015approximating,cole2017convex, anari2017nash}. In the following theorem, we state the guarantee in terms of the log-objective, and therefore, the guarantee becomes an additive one. 

\begin{theorem}\label{thm:convex-main}
    Let $(\cA, \cG, \mb{v}, \mb{w})$ be an instance of the weighted Nash Social Welfare problem with $\sum_{i\in \cA} w_i = 1$ and $|\cA| = n$ agents. There exists a polynomial time algorithm (Algorithm~\ref{alg:nsw-algo}) that, given $(\cA, \cG, \mb{v}, \mb{w})$, returns an assignment $\sigma: \cG \rightarrow \cA$ such that
    \begin{equation*}
        \nsw(\sigma) \geq \opt- 2\log{2} - \frac{1}{2e} - 2\cdot\kl{\mb{w}}{\mb{u}},
    \end{equation*}
    where $\opt$ is the optimal log-objective for the instance and $D_{KL}(\mb{w}||\mb{u})=\log n -\sum_{i\in \cA} w_i\log \frac{1}{w_i}.$
\end{theorem}

 Observe that the KL-divergence term $\kl{\mb{w}}{\mb{u}}=\left(\log n -\sum_{i\in \cA} w_i \log \frac{1}{w_i}\right)$ is always upper bounded by $\log(nw_{\max})$, which is exactly the guarantee of previous work~\cite{garg2023approximating}. In many settings, the term $2\cdot \kl{\mb{w}}{\mb{u}}$ can be significantly smaller than $nw_{\max}$. For example, consider the setting where $w_1=\frac{1}{\log n}$ and $w_i=\frac{1}{n-1}(1-\frac{1}{\log n})$ for $i=2,\ldots, n$, i.e., one agent has a significantly higher weight than the others. 
  Then \allowdisplaybreaks
    \begin{align*}
       \kl{\mb{w}}{\mb{u}} &= \frac{1}{\log{n}} \log\left(\frac{n}{\log{n}}\right) + \left(1-\frac{1}{\log{n}}\right) \log\left(\frac{n}{n-1} \left(1-\frac{1}{\log{n}}\right)\right) \\
       &= 1- \frac{\log\log{n}}{\log{n}} + \left(1-\frac{1}{\log{n}}\right) \log\left(\frac{n}{n-1} \right) + \left(1-\frac{1}{\log{n}}\right) \log \left(1-\frac{1}{\log{n}}\right)\\
       &\leq 1 + \log\left(\frac{n}{n-1} \right) \leq 2.
    \end{align*} 
In this case, our results imply an $O(1)$-approximation, while previous results imply an $O(\frac{n}{\log n})$-approximation.

Our algorithm relies on two mathematical programming relaxations for the weig\-hted Nash Social Welfare problem, both of which generalize the convex relaxation for the unweighted version \cite{cole2015approximating,cole2017convex,anari2017nash}. The first relaxation, \nameref{eq:b-non-cvx}, is non-convex but retains a lot of structural insights obtained for the convex relaxation in the symmetric version. We show that the same rounding algorithm as in the symmetric version~\cite{cole2015approximating} gives an $O(1)$-approximation for the weighted version when applied to a fractional solution of the non-convex program. Although \nameref{eq:b-non-cvx} can be rounded efficiently, unfortunately, we cannot solve this relaxation due to its non-convex nature. Now, the second mathematical programming relaxation, \nameref{eq:b-cvx}, comes to the rescue. This relaxation is convex and thus can be solved efficiently, but is challenging to round. Our algorithm solves the convex relaxation, then uses the non-convex relaxation to measure the change in objective as it processes the solution and eventually rounds to an integral assignment. The approximation factor of $\kl{\mb{w}}{\mb{u}}$ arises due to the difference in objective values of these two programs. Section~\ref{sec:techoverview} provides a technical overview of the properties of the two relaxations.

Before stating our second result, we describe the two previous convex programming relaxations for the unweighted Nash Social Welfare problem presented in \cite{cole2017convex} and \cite{anari2017nash}.

\textbf{Equivalence of Relaxations.} Building on the algorithm of \cite{cole2015approximating}, \cite{cole2017convex} introduced the following relaxation for the unweighted Nash Social Welfare problem.
\begin{align}
   \max_{b,q} \quad & \sum_{i \in \cA} \sum_{j \in \cG} b_{ij} \log\left( v_{ij} \right) - \sum_{j \in \cG} \left(\sum_{i\in \cA} b_{ij}  \right)\log\left(\sum_{i\in \cA} b_{ij} \right) \label{eq:cvx-sym} \tag{CVX-Unweighted}\\
    \mathrm{s.t.} \quad &  \sum_{j} b_{ij} = 1 \quad \forall i \in \cA\notag\\
      \quad & \sum_{i} b_{ij} \leq 1 \quad \forall j \in \cG \notag\\
   \quad & b_{ij} \geq 0 \quad \forall (i,j) \in \cA\times \cG .\notag
\end{align}

They showed that \eqref{eq:cvx-sym} is a convex relaxation of the Nash Social Welfare objective, and the prices used by the algorithm presented in \cite{cole2015approximating} can be obtained as dual variables of \eqref{eq:cvx-sym}. Interestingly, the convex relaxation is not in terms of the assignment variables. Indeed, given an optimal assignment $\sigma: \cG \rightarrow \cA$, the corresponding setting of the variables $b_{ij}$ is 
\begin{equation}
  \label{eq:b-assignment}  b_{ij} = \begin{cases}
        \frac{v_{ij}}{\sum_{k \in \sigma^{-1}(i)} v_{ik}} & \text{ if } \sigma(j) = i \\
        0 & \text{ otherwise}.
    \end{cases}
\end{equation}
One can verify that $\mb{b}$ satisfies all the constraints in~\eqref{eq:cvx-sym}, and its objective value is equal to the log of the geometric means of the valuations.

\cite{anari2017nash} presented a different convex programming relaxation\footnote{The program is concave in $x$ and convex in $\log(y)$. A change of variable $y_j \mapsto \exp(-z_j)$ gives a concave-convex program in $x$ and $z$.}, \eqref{eq:sym-poly}, for unweighted NSW. They showed that the objective of \eqref{eq:sym-poly} is a log-concave function in $x$ and convex in $\log y$ to obtain an $e$-approximation for unweighted NSW. 
\begin{align}
\label{eq:sym-poly}    \max_{\x\geq \bm{0}}\inf_{\mb{y} > \bm{0}} \quad & \sum_{i \in \cA} \log\left(\sum_{j \in \cG} x_{ij} \; v_{ij} \;y_j\right) \tag{LogConcave-Unweighted}\\
    \mathrm{s.t.} \quad &  \sum_{i\in \cA} x_{ij} = 1 \quad \forall j \in \cG\notag\\
      \quad & \prod_{j \in S} y_j \geq 1 \quad \forall S \in \binom{\cG}{n}. \notag
\end{align}

Here, $\binom{\cG}{n}$ denotes the collection of subsets of $\cG$ of size $n$, where $n = |\cA|$.

On the surface, the programs \eqref{eq:sym-poly} and \eqref{eq:cvx-sym}, and their corresponding rounding algorithm are quite different: \cite{cole2017convex} used intuition from economics and market equilibrium to both arrive at \eqref{eq:cvx-sym} and also to round it, while \cite{anari2017nash} uses the properties of log-concave polynomials to round \eqref{eq:cvx-sym}.

However, our next result shows that these two convex programs indeed optimize the same objective. 
\begin{theorem} \label{thm:sym-eq}
    The optimal values of \eqref{eq:sym-poly} and \eqref{eq:cvx-sym} are the same.
\end{theorem}
The proof of this theorem, presented in Appendix \ref{app:sym}, relies on a series of transformations using convex duality.

Besides providing a novel connection between two very different approaches to the unweighted problem, Theorem \ref{thm:sym-eq} is also vital to derive our main algorithm for weighted Nash Social Welfare. Independently generalizing either of these approaches to the weighted case is challenging: \cite{cole2017convex,cole2015approximating} use intuition from economics to arrive at \eqref{eq:cvx-sym}, and these concepts do not generalize to the weighted case. On the other hand, there is a natural convex generalization of \eqref{eq:sym-poly} for the weighted case\footnote{We present this generalization in Appendix \ref{sec:asymmetric}}, it is not log-concave, and therefore the machinery introduced in cannot be used to analyze it.

Our approach leverages the connection between \cite{cole2017convex} and \cite{anari2017nash} stated in Theorem \ref{thm:sym-eq} to derive a more natural convex relaxation of weighted Nash Social Welfare, given by \nameref{eq:b-cvx}. We provide more concrete details on this relationship and how it leads to the two programs for weighted NSW in Appendix \ref{sec:mathematical_programs}.

\subsection{Preliminaries}

\paragraph{KL-Divergence.} For two probability distributions $\mb{p},\mb{q}$ over the same discrete domain $\mathcal{X}$, the KL-divergence between $\mb{p}$ and $\mb{q}$ is defined as 
\begin{equation*}
    \kl{\mb{p}}{\mb{q}}=\sum_{x \in \mathcal{X}}p(x)\log \left(\frac{p(x)}{q(x)}\right).
\end{equation*}

It is well-known, via Gibb's inequality, that the KL-divergence between two distributions is non-negative and is zero if and only if $\mb{p}$ and $\mb{q}$ are identical.

We use this fact crucially in the following claim. 
\begin{claim} \label{cor:val}
    Given positive reals $z_1, \ldots, z_d$, for any $y_1, y_2, \ldots, y_d \geq 0$,
\begin{align*}
    \sum_{j=1}^d y_j \log\left(\sum_{j=1}^d z_j\right) -  \sum_{j=1}^d y_j \log\left(\sum_{j=1}^d y_j\right) \geq \sum_{j=1}^d y_j \log z_j - \sum_{j=1}^d y_j \log y_j.
\end{align*}
\end{claim}

\begin{proof}
Define vectors $\mb{y}=(y_1,\ldots, y_d)$ and $\mb{z}=(z_1,\ldots, z_d)$. Then $\widetilde{\mb{y}}=\frac{\mb{y}}{\|\mb{y}\|_1}$ and $\widetilde{\mb{z}}=\frac{{\mb{z}}}{\|\mb{z}\|_1}$ define two probability distributions on $[d]$. The inequality is equivalent to $\kl{\widetilde{\mb{y}}}{\widetilde{\mb{z}}}\geq 0$.
\end{proof}

Moreover, if $\mb{q}$ is the uniform distribution on $\mathcal{X}$ and $\mb{p}$ is an arbitrary distribution on the same domain, then 
\begin{equation*}
    \kl{\mb{p}}{\mb{q}}=\log |\mathcal{X}| -\sum_{x\in \mathcal{X}} p(x) \log \frac{1}{p(x)}.
\end{equation*}

\paragraph{Feasibility Polytope.} 
Consider a complete bipartite graph $G = (\cG \cup \cA, E)$ where $E$ contains an edge $(i,j)$ for each $i \in \cA$ and $j\in \cG$. Let $\mathcal{M}(\cA)$ denote the set of all matchings in $G$ of size $|\cA|$, i.e., matchings which have an edge incident to every vertex in $\cA$. The convex hull of $\mathcal{M}(\cA)$, denoted by $\mathcal{P}(\cA, \cG)$, is defined by the following polytope.

\begin{definition}[Feasibility Polytope]\label{def:matching-polytope}
For a set of $m$ indivisible items, $\cG$, and a set of $n$ agents, $\cA$, the feasibility polytope, denoted by  $\mathcal{P}(\cA, \cG)$, is defined as
\begin{align*}
   \mathcal{P}(\cA, \cG) := \left\{ \mb{b} \in \mathbb{R}_{\geq 0}^{|\cA| \times |\cG|} \,:\,
     \sum_{j \in \cG} b_{ij} = 1 \; \forall i \in \cA \;,
     \sum_{i \in \cA} b_{ij} \leq 1 \; \forall j \in \cG \right\}.
\end{align*}
The constraint $\sum_{j \in \cG} b_{ij} = 1$ is called the Agent constraint for agent $i$, and the constraint $\sum_{i \in \cA} b_{ij} \leq 1$ is referred to as the Item constraint for item $j$.
\end{definition}

We call $\mathcal{P}(\cA, \cG)$ the feasibility polytope of $(\cA, \cG)$ and will refer to points in $ \mathcal{P}(\cA, \cG)$ as either feasible points or solutions.
In the next section, we use $\mathcal{P}(\cA, \cG)$ to define the feasible regions for both mathematical programs.

\subsection{Technical Overview}\label{sec:techoverview}

\paragraph{Programs for Weighted NSW.} 

We introduce two mathematical programs, \nameref{eq:b-cvx} and \nameref{eq:b-non-cvx} below as relaxations of the weighted Nash Social Welfare objective.
By setting $\mb{b}$ to be the same value as~\eqref{eq:b-assignment}, it is natural to see that both programs are indeed relaxations. 
The first program \nameref{eq:b-cvx} is a convex program for any non-negative weights $\mb{w}$, whereas the second program is not a convex program when the weights are not identical.

\begin{figure}[h]
  \noindent\begin{minipage}[t]{.5\linewidth}
\begin{framed}
    \begin{align}
 \max_{\mb{b}}  \quad &  f_{\cvx}(\mb{b}) := \sum_{i \in \cA} \sum_{j \in \cG} w_i \, b_{ij} \log{v_{ij}}\notag\\
  &- \sum_{j \in \cG} \sum_{i\in \cA} w_i\, b_{ij} \log\left( \sum_{i\in \cA} w_i\, b_{ij}\right) \notag\\
  &+ \sum_{i\in \cA} w_i \log{w_i} \notag\\
    \mathrm{s.t.} \quad &  \sum_{j\in \cG} b_{ij} = 1 \quad \forall i \in \cA\notag\\
   \quad &  \sum_{i\in \cA} b_{ij} \leq 1 \quad \forall j \in \cG \notag\\
   \quad & b_{ij} \geq 0 \quad \forall (i,j) \in \cA\times \cG\notag
\end{align}
\end{framed}
\captionsetup{labelformat=empty}
\caption{\textup{(CVX-Weighted)}}
\label{eq:b-cvx}
\end{minipage}\ignorespaces
\noindent\begin{minipage}[t]{.5\linewidth}
\begin{framed}
    \begin{align}
   \max_{\mb{b}}  \quad &f_{\ncvx}(\mb{b}) := \sum_{i \in \cA} \sum_{j \in \cG} w_i \, b_{ij} \log{v_{ij}} \notag\\
    &- \sum_{j \in \cG} \sum_{i\in \cA} w_i\, b_{ij} \log\left( \sum_{i\in \cA}  b_{ij}\right) \notag\\
    & \phantom{+ \sum_{i\in \cA} w_i \log\left(w_i\right)}\notag\\
    \mathrm{s.t.} \quad &  \sum_{j\in \cG} b_{ij} = 1 \quad \forall i \in \cA\notag\\
   \quad &  \sum_{i\in \cA} b_{ij} \leq 1 \quad \forall j \in \cG \notag\\
   \quad & b_{ij} \geq 0 \quad \forall (i,j) \in \cA\times \cG\notag
\end{align}
\end{framed}
\captionsetup{labelformat=empty}
 \caption{\textup{(NCVX-Weighted)}}
 \label{eq:b-non-cvx}
\end{minipage}%
\end{figure}

\begin{lemma}\label{lem:relaxation}
    \nameref{eq:b-cvx} and \nameref{eq:b-non-cvx} are relaxations of the weighted Nash Social Welfare problem. Moreover, when the weights are symmetric, i.e., $w_i = 1/n$ for all $i \in \cA$, the programs \nameref{eq:b-cvx} and~\nameref{eq:b-non-cvx} are equivalent to the convex program \eqref{eq:cvx-sym}.
\end{lemma}
We formally prove Lemma \ref{lem:relaxation} in Appendix \ref{sec:proofs}.

Note that the constraints for both \nameref{eq:b-cvx} and \nameref{eq:b-non-cvx} are identical to the feasibility polytope $\mathcal{P}(\cA, \cG)$.

While analogous to \eqref{eq:cvx-sym}, \nameref{eq:b-cvx} does not inherit a crucial property of \eqref{eq:cvx-sym}, making \nameref{eq:b-cvx} challenging to round: optimal solutions of \nameref{eq:b-cvx} need not be acyclic. Furthermore, the integrality gap of \nameref{eq:b-cvx} is non-trivial even in the case when there are exactly $n$ items. To circumvent these issues, we use \nameref{eq:b-non-cvx} as an intermediate step in our rounding algorithm which has a desirable property: given a  point $\mb{b} \in \mathcal{P}(\cA, \cG)$, one can efficiently find another point $\widetilde{\mb{b}} \in \mathcal{P}(\cA, \cG)$ without decreasing the objective $f_{\ncvx}$ such that the graph formed by support of $\widetilde{\mb{b}}$ is a forest, as stated in the following lemma. We formally define the support graphs in Definition \ref{def:supp-graph}.

\paragraph{The two relaxations.} We observe that the objective of~\nameref{eq:b-cvx} is a concave function, and thus, it is a polynomial time tractable convex program. However, the objective of \nameref{eq:b-non-cvx} is not necessarily concave when the weights $w_i$ are not uniform. Despite this, \nameref{eq:b-non-cvx} still satisfies many desirable properties:

\begin{restatable}{lemma}{forest}
    \label{lem:forest}
Let $\overline{\mb{b}}$ be any feasible point in $\mathcal{P}(\cA, \cG)$. Then there exists an acyclic solution, $\mb{b}^{\mathrm{forest}}$, in the support of $\overline{\mb{b}}$ such that   \begin{equation*}
       f_{\ncvx}(\mb{b}^{\mathrm{forest}}) \geq f_{\ncvx}(\overline{\mb{b}}).
    \end{equation*}
Moreover, such a solution can be found in time polynomial in $|\cA|$ and $|\cG|$.
\end{restatable}

Next, we establish that one can efficiently round any feasible point whose support graph is a forest to an integral assignment. 

\begin{restatable}{theorem}{nonconvexmain}\label{thm:non-convex-main}
    For a Nash Social Welfare instance $(\cA, \cG, \mb{v}, \mb{w})$, given a vector $\mb{b}\in \mathcal{P}(\cA, \cG)$ such that the support of $\mb{b}$ is a forest, there exists a deterministic polynomial time algorithm (Algorithm~\ref{alg:non-convex-algo}) which returns an assignment $\sigma:\cG\rightarrow \cA$ such that
    \begin{equation*}
         \nsw(\sigma)  \geq f_{\cvx}(\mb{b}) -\kl{\mb{w}}{\mb{u}}- 2\log{2} - \frac{1}{2e} .
    \end{equation*}
\end{restatable}

We remark that our algorithm for rounding \nameref{eq:b-non-cvx} (Algorithm \ref{alg:non-convex-algo}) is the same as that in \cite{cole2015approximating}. However, our analysis is quite different. Rather than using ideas from market interpretations of the problem, we utilize properties of \nameref{eq:b-cvx} and \nameref{eq:b-non-cvx}, which generalize to both the unweighted and the weighted versions of the problem.

 Our analysis relies crucially on two facts: the relative stability of optimal points of \nameref{eq:b-cvx} and the interplay between the values of $f_{\cvx}$ and $f_{\ncvx}$. First, we establish that any optimal point of \nameref{eq:b-cvx} is relatively stable; the difference between the objective values of an optimal solution and any feasible solution is independent of the valuations $\mb{v}$ and, therefore, can be bounded effectively. Second, we show that for any feasible solution, the difference between $f_{\cvx}$ and $f_{\ncvx}$ is, at most, the KL-Divergence between the weights and the constant vector.

Our analysis uses this stability property along with the structure of the feasibility polytope to iteratively sparsify an optimal solution and obtain a matching between the agents and bundles of items while only losing a constant factor in the objective. It is worth noting that the first term in the objective $f_{\ncvx}$ (and $f_{\cvx}$) is linear in the variable $\mb{b}$. As the constraint set on $\mb{b}$ is a matching polytope, the solution that optimizes a linear objective would be a matching in which all agents receive exactly one item. While such a matching would be very suboptimal compared to $\opt$, our algorithm constructs an augmented graph containing a matching with a value comparable to $\opt$. The crux of our algorithm is to find a feasible vector in the matching polytope for which $f_{\ncvx}$ is close to $\opt$, and the additional non-linear terms in  $f_{\ncvx}$ are relatively small.

The remaining challenge to our approach is that \nameref{eq:b-non-cvx} is not a convex program, and therefore, we cannot efficiently find a global optima that maximizes $f_{\ncvx}$. However, we show that the objective of \nameref{eq:b-non-cvx} and \nameref{eq:b-cvx} differ by at most the $\kl{\mb{w}}{\mb{u}}$, as stated in the following lemma. We leverage this fact to initialize \nameref{eq:b-non-cvx} with the globally optimal solution of \nameref{eq:b-cvx} to obtain the approximation guarantee.
\begin{lemma} \label{lem:diff}
    For any $\mb{b} \in \mathcal{P}(\cA, \cG)$ and weights $w_1, \ldots, w_n > 0$ with $\sum_{i\in \cA} w_i = 1$, 
    \begin{equation*}
        0 \leq f_{\cvx}(\mb{b}) - f_{\ncvx}(\mb{b}) \leq  \kl{\mb{w}}{\mb{u}} = \log n -\sum_{i \in \cA} w_i \log\frac{1}{w_i}.
    \end{equation*}
\end{lemma}

We obtain our main result in Theorem \ref{thm:convex-main} by combining Lemma \ref{lem:forest}, Theorem~\ref{thm:non-convex-main}, and Lemma~\ref{lem:diff}.

\subsection{Related Work}
 The problem of finding the allocation that maximizes the Nash Social Welfare objective is an NP-hard problem, as was proven by \cite{nguyen2014computational}. Additionally, \cite{lee2017apx} showed that finding such an allocation is also APX-hard. 
From an algorithmic perspective, the first constant factor approximation for the unweighted version was provided in ~\cite{cole2015approximating} using analogies from market equilibrium. \cite{cole2017convex} provided an improved analysis of the algorithm from \cite{cole2015approximating} and introduced a convex programming relaxation. Using an entirely different approach, \cite{anari2017nash} also provided a constant factor approximation for the unweighted variant, where their analysis employed the theory of log-concave polynomials. The best-known approximation factor with linear valuations of $1.45$ is due to \cite{barman2018finding}, where they provide a pseudopolynomial-time algorithm that finds an allocation that is envy-free up to one good. Their algorithm is entirely combinatorial and runs in polynomial time when the valuations are bounded.

Another setting of interest is when the valuation of each agent is submodular instead of additive. For instance, \cite{garg2021approximating} gave a constant factor approximation algorithm for maximizing the unweighted Nash Social welfare function when the agents' valuations are Rado, a special subclass of submodular functions. In the weighted case, the approximation factor of this algorithm depends on the ratio of the maximum weight to the minimum weight.
\cite{li2022constant} provided a constant factor approximation algorithm for the unweighted case with submodular valuations. More recently, \cite{garg2023approximating} gave a local search-based algorithm to obtain an $O(nw_{\max})$-approximation for the weighted case and a $4$-approximation for the unweighted case with submodular valuations. 
Note that this $O(nw_{\max})$-approximation factor was also the previously best-known approximation for the weighted case, even when considering additive valuations.

\section{Approximation Algorithm}
Before describing our algorithm, we need the following definitions.
\begin{definition}[Support Graph]\label{def:supp-graph}
    For a vector $\mb{b} \in \mathcal{P}(\cA, \cG)$, the support graph of $\mb{b}$, denoted by $G_{\mathrm{supp}}(\mb{b})$, is a bipartite graph with vertex set $\cA \cup \cG$. For any $i \in \cA$ and $j \in \cG$, the edge $(i, j)$ belongs to the edge set of $G$ if and only if $b_{ij} > 0$.
\end{definition}

\begin{definition}[Acyclic Solution]\label{def:acyclic}
    A vector $\mb{b} \in \mathcal{P}(\cA, \cG)$ is called an acyclic solution if the support graph of $\mb{b}$,  $G_{\mathrm{supp}}(\mb{b})$, does not contain any cycles.
\end{definition}

For ease of notation, given any feasible point $\mb{b}\in \mathcal{P}(\cA,\cG)$, we use vector $\mb{q}\in \RR^{|\cG|}$ to denote the projection of $\mb{b}$ to $\cG$, i.e.,
\begin{equation*}
    q_j := \sum_{i\in \cA} b_{ij}
\end{equation*}
 for each $j\in \cG$. Since $\mb{q}$ is completely defined by $\mb{b}$, with abuse of notation, we will interchangeably use $\mathcal{P}(\cA, \cG)$ to denote feasible vectors $\mb{b}$ as well as $(\mb{b}, \mb{q})$. Similarly, we will use $ f_{\ncvx}(\mb{b},\mb{q})$ and $f_{\cvx}(\mb{b},\mb{q})$ to also denote the objective $f_{\ncvx}(\mb{b})$ and  $f_{\cvx}(\mb{b})$, respectively.
With a slight abuse of notation, we define
\begin{equation*}
    f_{\ncvx}(\mb{b},\mb{q}) :=  \sum_{i\in \cA} \sum_{j \in \cG} w_i \, b_{ij} \log{v_{ij}} - \sum_{i \in \cA} \sum_{j \in \cG} w_i\, b_{ij} \log {q_j}.
\end{equation*}
for any $\mb{b} \in \mathcal{P}(\cA, \cG)$ and its projection $\mb{q} \in \R^{|G|}$.

Our main algorithm, Algorithm \ref{alg:nsw-algo}, begins by finding the optimal solution $\overline{\mb{b}}$ to the convex program~\nameref{eq:b-cvx}. It then constructs another feasible point, $\mb{b}^{\mathrm{forest}}$, in support of $\overline{\mb{b}}$ such that the support graph of $\mb{b}^{\mathrm{forest}}$ is a forest and $f_{\ncvx}$ at $\mb{b}^{\mathrm{forest}}$ is at least $f_{\ncvx}$ at $\overline{\mb{b}}$. In the final step, the algorithm rounds $\mb{b}^{\mathrm{forest}}$ to an integral solution using Algorithm~\ref{alg:non-convex-algo}. Theorem~\ref{thm:non-convex-main} establishes a bound on the rounding error incurred during Algorithm~\ref{alg:non-convex-algo}.

\begin{algorithm2e}[H]
\DontPrintSemicolon
\caption{Approximation Algorithm for Weighted Nash Social Welfare}
\label{alg:nsw-algo}

   \textbf{Input}.  NSW instance $(\cA, \cG, \mb{v}, \mb{w})$\;
   $\overline{\mb{b}} \leftarrow$ optimal solution of~\nameref{eq:b-cvx}\; 
    $\overline{\mb{q}} \leftarrow$ vector in $\R^{|\cG|}$ with $\overline{q}_j = \sum_{i \in \cA} \overline{b}_{ij}$\;
     $(\mb{b}^{\mathrm{forest}}, \mb{q}^{\mathrm{forest}})  \leftarrow$ acyclic solution in support of $\overline{\mb{b}}$ such that $f_{\ncvx}(\mb{b}^{\mathrm{forest}}) \geq f_{\ncvx}(\overline{\mb{b}})$  \;
   $\sigma \leftarrow$ output of Algorithm \ref{alg:non-convex-algo} with input $(\cA, \cG, \mb{v}, \mb{w}, \mb{b}^{\mathrm{forest}}, \mb{q}^{\mathrm{forest}})$\;
      \textbf{Output.} $\sigma$
      
\end{algorithm2e}

Lemma \ref{lem:forest}, which we re-state below for the reader's convenience, guarantees the existence of $\mb{b}^{\mathrm{forest}}$, ensuring that the algorithm is well-defined. 
It is worth mentioning that for the unweighted case, the existence of an acyclic optimum was utilized by ~\cite{cole2015approximating, cole2017convex} for the convex program \eqref{eq:cvx-sym}. In the weighted setting, this structural property is not inherited by the convex program \nameref{eq:b-cvx} but by the non-convex program \nameref{eq:b-non-cvx}.

\forest*

\begin{proof}
Let $G_{\mathrm{supp}}(\mb{\bar{b}})$ contain a cycle $(i_0, j_0, i_1, \ldots, j_{\ell-1}, i_\ell)$ with $i_0 = i_\ell$, where $i_x \in \cA$ and $j_y \in \cG$. 
 The main idea is to modify the variables $\mb{\bar{b}}$ on this cycle while ensuring the value of $\mb{\bar{q}}$ does not change. If $\mb{\bar{q}}$ is fixed, then $f_{\ncvx}(\cdot,\mb{\bar{q}} )$  is linear in the input, and as a result, we can \emph{cancel} the cycle by considering the following vector.
Define $\boldsymbol{\delta} \in \R^{|\cA| \times |\cG|}$ with $\delta_{i_x j_x} := 1$ and $\delta_{i_{x+1} j_x} := -1$ for $x \in \{0, \ldots, \ell-1\}$, and $\delta_{ij} := 0$ otherwise.

Note that $\sum_{i \in \cA} \delta_{ij} = 0$ for any item $j$. As a result,  for each $j\in \cG$,
\begin{equation*}
    \sum_{i \in \cA} \bar{b}_{ij} + \varepsilon \delta_{ij} =  \sum_{i \in \cA} \bar{b}_{ij} = \bar{q}_j.
\end{equation*}
Therefore, the change in $f_{\ncvx}$ is given by
\begin{align*}
    f_{\ncvx}(\mb{\bar{b}}  + \varepsilon \boldsymbol{\delta}, \mb{\bar{q}} ) - f_{\ncvx}(\mb{\bar{b}} , \mb{\bar{q}} ) &= \sum_{i\in \cA} \sum_{j \in \cG} \varepsilon\, w_i\,\delta_{ij} \log v_{ij} - \sum_{i\in \cA} \sum_{j \in \cG} \varepsilon\, w_i\,  \delta_{ij} \log{\bar{q}_j} := \varepsilon\; h(\boldsymbol{\delta}, \mb{\bar{q}}).
\end{align*}
Note that $h(\boldsymbol{\delta}, \mb{\bar{q}})$ is a linear function in $\boldsymbol{\delta}$.
So, if $h(\boldsymbol{\delta}, \mb{\bar{q}}) > 0$, then setting $\varepsilon = \max_{x} b_{i_{x+1} j_x}$ ensures that $ f_{\ncvx}(\mb{\bar{b}}  + \varepsilon \boldsymbol{\delta}, \mb{\bar{q}}) \geq f_{\ncvx}(\mb{\bar{b}}, \mb{\bar{q}}) $, and $\mb{\bar{b}}  + \varepsilon \boldsymbol{\delta} \in \mathcal{P}(\cA, \cG)$. In addition, the number of cycles in $G_{\mathrm{supp}}(\mb{\bar{b}}  + \varepsilon \boldsymbol{\delta})$ is strictly less than the number of cycles in $G_{\mathrm{supp}}(\mb{\bar{b}} )$.

Similarly, if $h(\boldsymbol{\delta}, \mb{\bar{q}}) \leq 0$, setting $\varepsilon =- \max_{x} b_{i_{x} j_x}$ gives the same guarantees. Iterating this cycle canceling process until the support contains no cycles leads to the required solution.
\end{proof}

By combining Lemma \ref{lem:forest} with Lemma \ref{lem:diff}, we obtain the following corollary.
\begin{corollary}\label{cor:acyclic}
Let $\overline{\mb{b}}$ be any feasible point in $\mathcal{P}(\cA, \cG)$. Then, there exists an acyclic solution, $\mb{b}^{\mathrm{forest}}$, in the support of $\overline{\mb{b}}$ such that   
    \begin{equation*}
       f_{\cvx}(\mb{b}^{\mathrm{forest}}) \geq f_{\cvx}(\overline{\mb{b}})- \kl{\mb{w}}{\mb{u}}.
    \end{equation*}
Moreover, such a $\mb{b}^{\mathrm{forest}}$ can be found in time polynomial in $|\cA|$ and $|\cG|$.
\end{corollary}

Before presenting Algorithm \ref{alg:non-convex-algo}, we give the proof of Theorem~\ref{thm:convex-main}, which now follows directly from Theorem~\ref{thm:non-convex-main} and Corollary~\ref{cor:acyclic}, as outlined below.

\begin{proof}[Proof of Theorem~\ref{thm:convex-main}] 
Let $(\overline{\mb{b}}, \overline{\mb{q}})$ and $(\mb{b}^\mathrm{forest}, \mb{q}^\mathrm{forest})$ denote the feasible points defined in Step 1 and Step 3 of Algorithm~\ref{alg:nsw-algo}, respectively. Let $\sigma^{\star}$ be the assignment returned by 
Algorithm~\ref{alg:non-convex-algo} on input $(\mb{b}^{\mathrm{forest}}, \mb{q}^{\mathrm{forest}})$. By Theorem~\ref{thm:non-convex-main}, we have
\begin{align*}
    \nsw(\sigma^\star) &\geq  f_{\cvx}(\mb{b}^{\mathrm{forest}}, \mb{q}^{\mathrm{forest}}) - \kl{\mb{w}}{\mb{u}} - 2\log{2}-\frac{1}{2e}\\
    &\stackrel{(i)}{\geq} f_{\cvx}(\mb{\overline{b}}, \mb{\overline{q}}) - 2\cdot \kl{\mb{w}}{\mb{u}} - 2\log{2}-\frac{1}{2e}\\
     &\stackrel{(ii)}{\geq} \mathrm{OPT} - 2\cdot \kl{\mb{w}}{\mb{u}} - 2\log{2}-\frac{1}{2e}.
\end{align*}
Here, $(i)$ follows from Corollary~\ref{cor:acyclic} and $(ii)$ follows from Lemma~\ref{lem:relaxation}. 
\end{proof}

\subsection{Rounding an Acyclic Solution}
Given an acyclic solution $\mb{b}$, Algorithm \ref{alg:non-convex-algo} returns an assignment, $\sigma^\star$, such that $\nsw(\sigma^\star)$ is comparable to $f_{\cvx}(\mb{b})$, as stated in Theorem \ref{thm:non-convex-main}.

\begin{algorithm2e}
\DontPrintSemicolon
\caption{Algorithm for Rounding an Acyclic Solution}
\label{alg:non-convex-algo}
     \textbf{Input}. NSW instance $(\cA, \cG, \mb{v}, \mb{w})$, acyclic solution $(\mb{b},\mb{q}) \in \mathcal{P}(\cA, \cG)$ \;
     $(\mb{b}^\star, \mb{q}^\star)   \leftarrow$ optimal solution of~\nameref{eq:b-cvx} restricted to the support of $(\mb{b}, \mb{q}) $\;
    $ F^\star \leftarrow$ $G_{\mathrm{supp}}(\mb{b}^\star )$ with every tree rooted at an agent node\;
    $\widetilde{F}\leftarrow$ Forest obtained by removing edges between item $j$ and its children in $F^\star$ whenever $q_j^\star < 1/2$ \tcc*{pruning step}
       $L_i^{\star} \leftarrow$ set of leaf children of agent $i$ in $\widetilde{F}$ and let $L^\star = \cup_i \{L_i^\star\}$\;
       $M^\star \leftarrow$  matching between $\cA \rightarrow \cG \backslash L^\star$ in $\widetilde{F}$ which maximizes weight function
       $w_{\widetilde{F}}(M) := \sum_{i \in \cA} w_i \, \log\left(v_{i M(i)}+\sum_{j \in L^\star_i} v_{ij} \right)$ \footnotemark   \;
       $\sigma^\star \leftarrow$ assignment of $\cG$ to $\cA$ with $\sigma^\star(j) = i$ if $j\in \{L^\star_i \cup M^\star(i)\}$ \tcc*{matching step}
  \textbf{Output.} $\sigma^\star$
\end{algorithm2e}
\footnotetext{If agent $i$ in unmatched in $M$, we let $v_{iM(i)} = 0$}
In the first step, Algorithm \ref{alg:non-convex-algo} finds an optimal solution, denoted by $\mb{b}^\star$, to~\nameref{eq:b-cvx} restricted to the support of $\mb{b}$, i.e., $\mb{b}^\star$ is the optimal solution to~\nameref{eq:b-cvx} on input $(\cA, \cG, \widetilde{\mb{v}},\mb{w})$, where $\widetilde{v}_{ij} = 0$ if $b_{ij} = 0$, and $\widetilde{v}_{ij} = {v}_{ij}$ otherwise.
This step is crucial as it allows us to utilize the stability properties of stationary points of~\nameref{eq:b-cvx}.

Next, the algorithm implements a ``pruning'' step to sparsify $\mb{b}^\star$: it removes edges between any item with $q^\star _j < 1/2$ and its children in $F^\star$. Here, $F^\star$ is the support graph of $\mb{b}^\star$ with every tree rooted at agent nodes. This step is equivalent to assigning each item $j$ with $q^\star _j < 1/2$ to its parent agent in $F^\star$. As a result, any item with $q_j^\star < 1/2$ is a leaf in the pruned forest, $\widetilde{F}$. Since removing edges will exclude certain items from being assigned to some agents, pruning can lead to a sub-optimal solution. We bound this loss in objective by showing the existence of a fractional solution $(\mb{b}^{\mathrm{pruned}}, \mb{q}^{\mathrm{pruned}})$ whose support graph is a subset of the pruned forest, $\widetilde{F}$, and $f_{\cvx}(\mb{b}^{\mathrm{pruned}})$ is comparable to $f_{\cvx}(\mb{b}^{\star})$. For concrete details, see Section \ref{sec:rounding}.

It is important to emphasize that the algorithm does not need to find $(\mb{b}^{\mathrm{pruned}}, \mb{q}^{\mathrm{pruned}})$. The mere existence of $(\mb{b}^{\mathrm{pruned}}, \mb{q}^{\mathrm{pruned}})$ is enough to guarantee that the assignment returned by the algorithm will be good, as explained below. 

After the pruning step, the algorithm assigns every leaf item in the pruned forest to its parent. We use $L_i^\star$ to denote the set of leaf items whose parent is agent $i$ and $L^\star =\cup_{i \in \cA} L_i^\star$ to denote the set of all leaf items in the pruned forest. So, each agent $i$ receives all the items in the bundle $L_i^\star$. In the matching step, the algorithm assigns at most one additional item to each agent by finding a maximum weight matching between agents $\cA$ and items $\cG \backslash L^\star$ (the set of non-leaf items in the pruned forest). This matching is determined using an augmented weight function, denoted by $w_{\widetilde{F}}$. The weight of a matching $M$ between $\cA$ and $\cG\setminus L^\star$ in the pruned forest is defined as follows:
\begin{equation*}
     w_{\widetilde{F}}(M) :=  \sum_{i \in \cA} w_i \, \log\left(v_{iM(i)} + \sum_{j \in L^\star_i} v_{ij}  \right)\,,
\end{equation*}
where $v_{iM(i)} = 0$ if $i$ is not matched in $M$. Observe that this weight function exactly captures the weighted Nash Social Welfare objective when agent $i$ is assigned the item set $S_i:= \{M(i)\cup L^\star_i\}$ for each $i\in \cA$. Moreover, finding the optimal matching $M$ can be easily formulated as a maximum weight matching problem in a bipartite graph.

Since the standard linear programming relaxation for the bipartite matching problem is integral, it is enough to demonstrate the existence of a \emph{fractional matching} with a large weight $w_{\widetilde{F}}$ in the pruned forest. In Section \ref{sec:matching}, we show how to construct a fractional matching corresponding to $\mb{b}^{\mathrm{pruned}}$, such that the weight of this matching is comparable to the objective $f_{\ncvx}(\mb{b}^{\mathrm{pruned}})$. We emphasize that this matching corresponding to $\mb{b}^{\mathrm{pruned}}$ is only required for the sake of analysis: to lower bound the performance of the matching returned by the algorithm. We do not need to know $\mb{b}^{\mathrm{pruned}}$ for the execution of the algorithm.

\section{Rounding via the Non-Convex Relaxation} \label{sec:rounding}

In this section, we prove Theorem~\ref{thm:non-convex-main} by establishing properties of support-restricted optimal solutions of \nameref{eq:b-cvx}. 
First, in Lemma \ref{lem:after-pruning}, we show that any optimum whose support is restricted to a forest can be ``pruned'' to a feasible solution while only losing a constant factor in the objective. Specifically, we show that given a support restricted optimum $(\mb{b}^\star , \mb{q}^\star )$, we
can construct a feasible solution $(\mb{b}^{\mathrm{pruned}} , \mb{q}^{\mathrm{pruned}} )$ such that any item with $q^{\mathrm{pruned}}_j < 1/2$ is a leaf in support graph of $\mb{b}^{\mathrm{pruned}}$, and $ f_{\cvx}(\mb{b}^{\mathrm{pruned}}, \mb{q}^{\mathrm{pruned}}) \geq f_{\cvx}(\mb{b}^\star , \mb{q}^\star ) - \log{2}$.

Second, in Lemma \ref{lem:fractional}, we demonstrate the existence of a matching in the support graph of $\mb{b}^{\mathrm{pruned}}$ such that the augmented weight function of this matching differs from $f_{\ncvx}(\mb{b}^{\mathrm{pruned}})$ by a constant factor. After presenting these two lemmas, we provide the proof of Theorem \ref{thm:non-convex-main}.

\begin{restatable}{lemma}{pruning}
\label{lem:after-pruning}
     Let $(\mb{b}^\star , \mb{q}^\star )$ be the optimal solution of \nameref{eq:b-cvx} in the support of some acyclic feasible point $\mb{b}^{\mathrm{forest}}$. Let $F$ be a directed forest formed by $G_{\mathrm{supp}}(\mb{b}^\star)$ when every tree is rooted at an agent node. Then, there exists an acyclic feasible point $(\mb{b}^{\mathrm{pruned}}, \mb{q}^{\mathrm{pruned}})$ in $\mathcal{P}(\cA, \cG)$ such that $G_{\mathrm{supp}}(\mb{b}^{\mathrm{pruned}})$ is a subgraph of $G_{\mathrm{supp}}(\mb{b}^\star)$ and
     \begin{itemize}
         \item $q_j^{\mathrm{pruned}} \geq q_j^\star$ for any item $j$ with $q_j^\star \geq 1/2$,
         \item each item with $q_j^\star < 1/2$ is a leaf in $G_{\mathrm{supp}}(\mb{b}^{\mathrm{pruned}})$ connected to its parent in $F$, and
        \item $ f_{\cvx}(\mb{b}^{\mathrm{pruned}}, \mb{q}^{\mathrm{pruned}}) \geq f_{\cvx}(\mb{b}^\star , \mb{q}^\star ) - \log{2}.$
     \end{itemize}
\end{restatable} 
The proof of Lemma~\ref{lem:after-pruning} relies on the stability properties of optimal solutions of \nameref{eq:b-cvx}, as outlined in Section~\ref{sec:local}.

\begin{lemma}\label{lem:fractional}
 Let $(\mb{b}, \mb{q})$ be an acyclic solution in $\mathcal{P}(\cA, \cG)$ such that every item with $q_j < 1/2$ is a leaf in $G_{\mathrm{support}}(\mb{b})$. Let $S: \cA \rightarrow 2^{\cG}$ be a function such that for each agent $i$, $S(i)$ is a subset of the leaf items connected to agent $i$ in $G_{\mathrm{supp}}(\mb{b})$, and  $S(i)$ contains all children of agent $i$ with $q_j < 1/2$. 
 Then, there exists a matching $M$ in $G_{\mathrm{supp}}(\mb{b})$ between the vertices in $\cA$ and $\{\cG\backslash \cup_i\{S(i)\} \}$ such that
\begin{align*}
\sum_{i \in \cA} w_i \log\left(v_{iM(i)} + \sum_{j \in S(i)} v_{ij}\right) &\geq   f_{\ncvx}(\mb{b},\mb{q}) -\log 2 -\frac{1}{2e}\,, 
\end{align*}
where $v_{iM(i)} = 0$ if agent $i$ is not matched in $M$.
\end{lemma}

We prove this lemma in Section \ref{sec:matching}.

\begin{proof}[Proof of Theorem \ref{thm:non-convex-main}]
Given $(\mb{b}, \mb{q})$ such that $G_{\mathrm{supp}}(\mb{b})$ is a forest, let $(\mb{b}^\star, \mb{q}^\star)$ be the optimal solution of~\nameref{eq:b-cvx} restricted to support of $\mb{b}$, let $\widetilde{F}$ denote the forest obtained after pruning $G_{\mathrm{supp}}(\mb{b}^\star)$. Let $L_i^\star$ denote the set of leaf children of agent $i$ in $\widetilde{F}$.

Let $(\mb{b}^{\mathrm{pruned}}, \mb{q}^{\mathrm{pruned}})$ be a feasible solution guaranteed by Lemma~\ref{lem:after-pruning} on input $(\mb{b}^\star, \mb{q}^\star)$. Since Lemma~\ref{lem:after-pruning} guarantees that $G_{\mathrm{supp}}(\mb{b}^{\mathrm{pruned}})$ is a subset of $G_{\mathrm{supp}}(\mb{b}^\star)$, and every item with $q_j^\star < 1/2$ is a leaf in  $G_{\mathrm{supp}}(\mb{b}^{\mathrm{pruned}})$, we conclude that  $G_{\mathrm{supp}}(\mb{b}^{\mathrm{pruned}})$ is a subgraph of $\widetilde{F}$. 

In addition, for any agent $i$, $L_i^\star$ is a subset of the leaf children of $i$ in $G_{\mathrm{supp}}(\mb{b}^{\mathrm{pruned}})$ as $G_{\mathrm{supp}}(\mb{b}^{\mathrm{pruned}})$ is a subgraph of $\widetilde{F}$. Furthermore, if $q_j^{\mathrm{pruned}} < 1/2$, then we claim that $j$ is a leaf in $G_{\mathrm{supp}}(\mb{b}^{\mathrm{pruned}})$ with parent $i$ such that $j \in L_i^\star$ in $\widetilde{F}$. Since $q_j^{\mathrm{pruned}} < 1/2$, by the first point of Lemma~\ref{lem:after-pruning}, we have $q_j^\star < 1/2$. As a result, item $j$ is a leaf in $G_{\mathrm{supp}}(\mb{b}^{\mathrm{pruned}})$ connected to its parent in $\widetilde{F}$. So, item $j$ would be pruned in $\widetilde{F}$, and therefore, by definition, $j \in L_i^\star$.

Therefore, for each agent $i$, the set $L_i^\star$ is a subset of the set of leaves of agent $i$ in $G_{\mathrm{supp}}(\mb{b}^{\mathrm{pruned}})$, and $L_i^\star$ contains all the items with $q_j^{\mathrm{pruned}} < 1/2$ in $G_{\mathrm{supp}}(\mb{b}^{\mathrm{pruned}})$. So, the function $S(i) = L_i^\star$ satisfies the constraints of Lemma~\ref{lem:fractional} with input $(\mb{b}^{\mathrm{pruned}}, \mb{q}^{\mathrm{pruned}})$.

Using Lemma~\ref{lem:fractional} on $(\mb{b}^{\mathrm{pruned}}, \mb{q}^{\mathrm{pruned}})$ with function $S(i) = L_i^\star$, we conclude that there exists a matching, $M$, in $G_{\mathrm{supp}}(\mb{b}^{\mathrm{pruned}})$ such that
\begin{align*}
     \sum_{i \in \cA} w_i \log\left(v_{iM(i)} + \sum_{j \in L^\star_i} v_{ij}\right) &= \sum_{i \in \cA} w_i \log\left(v_{iM(i)} + \sum_{j \in S(i)} v_{ij}\right)\\
     &\geq   f_{\ncvx}(\mb{b}^{\mathrm{pruned}}, \mb{q}^{\mathrm{pruned}}) -\log 2 -\frac{1}{2e}.
\end{align*}

Since $G_{\mathrm{supp}}(\mb{b}^{\mathrm{pruned}})$ is a subgraph of $\widetilde{F}$, the matching $M$ is also present in $\widetilde{F}$. Therefore, the matching $M^\star$ (and corresponding assignment $\sigma^\star$) returned by Algorithm~\ref{alg:non-convex-algo} satisfies \allowdisplaybreaks
\begin{align*}
    \nsw(\sigma^\star) &= \sum_{i \in \cA} w_i \log\left(v_{iM^\star(i)} + \sum_{j \in L_i^\star} v_{ij}\right) \stackrel{(i)}{\geq} \sum_{i \in \cA} w_i \log\left(v_{iM(i)} + \sum_{j \in L_i^\star} v_{ij}\right) \\
    &\stackrel{(ii)}{\geq} f_{\ncvx}(\mb{b}^{\mathrm{pruned}}, \mb{q}^{\mathrm{pruned}}) -\log 2 -\frac{1}{2e} \\
    &\stackrel{(iii)}{\geq} f_{\cvx}(\mb{b}^{\mathrm{pruned}}, \mb{q}^{\mathrm{pruned}}) -\kl{\mb{w}}{\mb{u}} -\log 2 -\frac{1}{2e} \\
      &\stackrel{(iv)}{\geq} f_{\cvx}(\mb{b}^\star, \mb{q}^\star) -\kl{\mb{w}}{\mb{u}}-2\log{2} - \frac{1}{2e}\,\\
       &\stackrel{(v)}{\geq} f_{\cvx}(\mb{b}, \mb{q}) -\kl{\mb{w}}{\mb{u}} -2\log{2} - \frac{1}{2e}.
\end{align*}
Here, $(i)$ follows from the optimality of $M^\star$, $(ii)$ follows from Lemma~\ref{lem:fractional}, $(iii)$ follows from Lemma \ref{lem:diff}, $(iv)$ follows from Lemma \ref{lem:fractional}, and $(v)$ follows from the optimality of $\mb{b}^{\star}$.
\end{proof}

\subsection{Pruning Small Items}\label{sec:local} \label{sec:pruning}

In this section, we prove Lemma \ref{lem:after-pruning} by establishing some properties of the set of (support restricted) optimal solutions of \nameref{eq:b-cvx} in Lemma \ref{lem:a-change} and Lemma \ref{lem:pruned-sol}. 

First, we show that any optimal solution of \nameref{eq:b-cvx} is relatively stable, i.e., the change in function value when moving away from the optimal solution can be quantified in terms of how much we deviate from that solution. 
We formalize the stability property as follows.

\begin{restatable}{lemma}{a-change}\label{lem:a-change}
  Let $(\mb{b}^\star , \mb{q}^\star )$ be the optimal solution of \nameref{eq:b-cvx} in the support of some acyclic feasible point $\mb{b}^{\mathrm{forest}}$. Let $(\mb{b},\mb{q})$ be a feasible point in $\mathcal{P}(\cA, \cG)$ such that the support of $\mb{b}$ is a subset of the support of $\mb{b}^\star $, and for any $j \in \cG$, if $q^\star _j = 1$, then $q_j = 1$. Then \begin{equation*}
        f_{\cvx}(\mb{b}^\star , \mb{q}^\star ) - f_{\cvx}(\mb{b},\mb{q}) =  \sum_{j \in \cG} \sum_{i\in \cA} w_i\, b_{ij}  \log{\left(\frac{\sum_{i\in \cA} w_i\, b_{ij} }{\sum_{i\in \cA} w_i\, b_{ij}^\star }\right)}.
    \end{equation*}
\end{restatable}
We provide the proof of this lemma in Appendix \ref{sec:proofs}.

Second, in Lemma \ref{lem:pruned-sol}, we show that any acyclic optimal solution of \nameref{eq:b-cvx} can be pruned to a feasible solution, denoted by $\mb{b}^{\mathrm{pruned}}$, which is amenable to rounding. Specifically, we show that given a first-order stationary point $(\mb{b^\star}, \mb{q}^\star)$, we can construct a feasible solution $(\mb{b}^{\mathrm{pruned}}, \mb{q}^{\mathrm{pruned}})$ such that any item with $q_j^{\mathrm{pruned}} < 1/2$ is a leaf in support of $\mb{b}^{\mathrm{pruned}}$ and $b^{\mathrm{pruned}}_{ij} \leq \min \{1, 2b^\star_{ij} \}$ for any agent $i$ and item $j$.

\begin{lemma} \label{lem:pruned-sol}
 Let $(\mb{b}^\star, \mb{q}^\star)$ be an acyclic feasible point in $\mathcal{P}(\cA, \cG)$. Let $F$ be a directed forest formed by $G_{\mathrm{supp}}(\mb{b}^\star)$ when every tree is rooted at an arbitrary agent node. Then, there exists a feasible solution $(\mb{b}^{\mathrm{pruned}}, \mb{q}^{\mathrm{pruned}})$  such that $G_{\mathrm{supp}}(\mb{b}^{\mathrm{pruned}})$ is a subgraph of $G_{\mathrm{supp}}(\mb{b}^*)$,
 \begin{itemize}
    \item $q_j^\star \leq q_j^{\mathrm{pruned}}$ for each item $j$ with $q_j^\star \geq 1/2$, 
     \item each item with $q_j^\star < 1/2$ is a leaf in $G_{\mathrm{supp}}(\mb{b}^{\mathrm{pruned}})$ connected to its parent in $F$, and
     \item for any $(i,j) \in \cA \times \cG$, $ b^{\mathrm{pruned}}_{ij} \leq \min\{1, 2\cdot b^\star_{ij}\} .$

 \end{itemize}
 
\end{lemma}

Before proving Lemma \ref{lem:pruned-sol}, we use Lemma \ref{lem:pruned-sol} along with Lemma \ref{lem:a-change} to prove Lemma \ref{lem:after-pruning}.

\begin{proof}[Proof of Lemma~\ref{lem:after-pruning}]
   By Lemma~\ref{lem:pruned-sol}, there exists a feasible solution 
   $(\mb{b}^{\mathrm{pruned}},\mb{q}^{\mathrm{pruned}})$ such that the support graph, $G_{\mathrm{supp}}(\mb{b}^{\mathrm{pruned}})$, is a subgraph of $G_{\mathrm{supp}}(\mb{b})$ and $(\mb{b}^{\mathrm{pruned}},\mb{q}^{\mathrm{pruned}})$ satisfies the first two items claimed in the lemma. Furthermore, for any $(i,j) \in \cA \times \cG$, $ b^{\mathrm{pruned}}_{ij} \leq \min\{1, 2\cdot b^\star_{ij}\}$.

So using Lemma~\ref{lem:a-change}, the difference in objective between $(\mb{b}^\star , \mb{q}^\star )$ to $(\mb{b}^{\mathrm{pruned}},\mb{q}^{\mathrm{pruned}})$ is bounded as follows
    \begin{align*}
        f_{\cvx}(\mb{b}^\star , \mb{q}^\star ) - f_{\cvx}(\mb{b}^{\mathrm{pruned}},\mb{q}^{\mathrm{pruned}}) =  \sum_{j \in \cG} \sum_{i\in \cA} w_i\, b^{\mathrm{pruned}}_{ij}  \log{\left(\frac{\sum_{i\in \cA} w_i\, b^{\mathrm{pruned}}_{ij} }{\sum_{i\in \cA} w_i\, b_{ij}^\star }\right)} .
    \end{align*}

Since $b^{\mathrm{pruned}}_{ij} \leq \min\{1, \, 2\,b^\star_{ij}\}$ for each $(i,j)$, we have $\sum_i w_i\, b^{\mathrm{pruned}}_{ij} \leq 2\sum_i w_i \,b^\star_{ij}$.
 \begin{align}
        f_{\cvx}(\mb{b}^\star , \mb{q}^\star ) - f_{\cvx}(\mb{b}^{\mathrm{pruned}},\mb{q}^{\mathrm{pruned}}) \leq \sum_{j\in \cG} \sum_{i\in \cA} w_i\, b^{\mathrm{pruned}}_{ij} \log{2}. \label{eq:3.5.1}
\end{align}
The feasibility of $\mb{b}^{\mathrm{pruned}}$ implies
\begin{equation*}
     \sum_{j\in \cG} \sum_{i\in \cA} w_i\, b^{\mathrm{pruned}}_{ij} =  \sum_{i\in \cA} w_i \sum_{j\in \cG} b^{\mathrm{pruned}}_{ij} = \sum_{i\in \cA} w_i = 1.
\end{equation*}
Plugging this bound in equation~\eqref{eq:3.5.1} completes the proof.
\end{proof}
\begin{proof}[Proof of Lemma \ref{lem:a-change}] 
If $\mb{b}^\star$ is an optimal solution of \nameref{eq:b-cvx}, then using the KKT conditions, there exist real numbers $\lambda_i$ for each $i \in \cA$, $\eta_j \geq 0$ for each $j \in \cG$, and $\alpha_{ij}\geq 0$ for every $(i,j) \in \cA \times \cG$ such that
\begin{align*}
    &\frac{\partial L}{\partial b^\star _{ij}} = w_i \log{v_{ij}}  - w_i - w_i\log\left(\sum_{i \in \cA} w_i b_{ij}^\star\right)-\lambda_i-\eta_j +\alpha_{ij} = 0.
\end{align*}

In addition, by complementary slackness, we have $\eta_j(1- \sum_{i \in \cA } b_{ij}^\star) = 0$ for each item $j$ and $\alpha_{ij}  b^\star _{ij} = 0$ for each $(i,j) \in \cA \times \cG$. Using these complementary slackness conditions, if $b^\star _{ij} > 0$, then
\begin{equation}
     w_i\log{v_{ij}} = w_i +w_i\log\left(\sum_{i \in \cA} w_i b_{ij}^\star\right)+\lambda_i+\eta_j .\label{eq:10} 
\end{equation}

Now, expanding the difference between the two function values, we get
    \begin{align}
         f_{\cvx}(\mb{b}^\star , \mb{q}^\star ) - f_{\cvx}(\mb{b},\mb{q}) &= \sum_{i\in \cA} \sum_{j \in \cG}\left(b^\star _{ij} - b_{ij}\right) \cdot  w_i \log{v_{ij}} -  \sum_{j \in \cG} \sum_{i\in \cA} w_i \, b^\star _{ij} \log{\left(\sum_{i\in \cA} w_i \, b^\star _{ij}\right)} \notag\\
         &\quad + \sum_{j \in \cG} \sum_{i\in \cA} w_i \, b_{ij} \log{\left(\sum_{i\in \cA} w_i \, b_{ij}\right)} . \label{eq:difference}
    \end{align}

Substituting the value of $v_{ij}$ from equation \eqref{eq:10} in equation~\eqref{eq:difference} gives
     \begin{align*}
         f_{\cvx}(\mb{b}^\star , \mb{q}^\star ) - f_{\cvx}(\mb{b},\mb{q}) &= \sum_{i\in \cA} \sum_{j \in \cG} (b^\star _{ij} - b_{ij}) \left(w_i\log{\left(\sum_{i\in \cA} w_i \, b^\star _{ij}\right)} + \lambda_i + w_i + \eta_j\right) \\
         & \quad -  \sum_{j \in \cG} \sum_{i\in \cA} w_i \, b^\star _{ij} \log{\left(\sum_{i\in \cA} w_i \, b^\star _{ij}\right)} + \sum_{j \in \cG} \sum_{i\in \cA} w_i \, b_{ij} \log{\left(\sum_{i\in \cA} w_i \, b_{ij}\right)}\\
         &= \sum_{j \in \cG} \sum_{i\in \cA} w_i\, b_{ij}  \log{\left(\frac{\sum_{i\in \cA} w_i\, b_{ij} }{\sum_{i\in \cA} w_i\, b_{ij}^\star }\right)}  + \sum_{i \in \cA} \left(\lambda_i + w_i\right) \left(\sum_{j \in \cG} b^\star _{ij} -\sum_{j \in \cG}  b_{ij}\right) \\
         &\quad + \sum_{j \in \cG} \eta_j  \left(\sum_{i \in \cA} b^\star _{ij} -\sum_{i \in \cA}  b_{ij}\right).
    \end{align*}
Using $\sum_{j\in \cG} b_{ij} = \sum_{j\in \cG} b^\star _{ij} = 1$ for every $i \in \cA$, we get
\begin{align*}
    f_{\cvx}(\mb{b}^\star , \mb{q}^\star ) -f_{\cvx}(\mb{b},\mb{q}) &=  \sum_{j \in \cG} \sum_{i\in \cA} w_i\, b_{ij}  \log{\left(\frac{\sum_{i\in \cA} w_i\, b_{ij} }{\sum_{i\in \cA} w_i\, b_{ij}^\star }\right)}  + \sum_{j \in \cG} \eta_j  \left( q^\star_j -q_j\right),
\end{align*}
where the last equation follows from the definitions of $q_j$ and $q^\star _j$.

Note that by complementary slackness, $\eta_j (1-q^\star_j) = 0$ for any $j \in \cG$. So if $q^\star _j < 1$, then $\eta_j = 0$ and therefore $\eta_j (q^\star _j - q_j) = 0$. If $q^\star _j = 1$, then by the hypothesis of the Lemma, $q_j = 1$, and again we obtain that $\eta_j (q^\star _j - q_j) = 0$. Using this bound in the above equation gives
\begin{align*}
        f_{\cvx}(\mb{b}^\star , \mb{q}^\star ) -f_{\cvx}(\mb{b},\mb{q}) &=  \sum_{j \in \cG} \sum_{i\in \cA} w_i\, b_{ij}  \log{\left(\frac{\sum_{i\in \cA} w_i\, b_{ij} }{\sum_{i\in \cA} w_i\, b_{ij}^\star }\right)}.
    \end{align*}
\end{proof}

Before proving Lemma \ref{lem:pruned-sol}, we need the following lemma about the feasibility of a solution when we decrease the $b_{ij}$ for some edge $(j,i)$ in the support forest of $\mb{b}$.
\begin{restatable}{lemma}{redistribution}
    \label{lem:redistribution}
    Let $(\mb{b},\mb{q})$ be an acyclic feasible point in $\mathcal{P}(\cA, \cG)$, and let $F$ be a directed forest formed by $G_{\mathrm{supp}}(\mb{b})$ when every tree is rooted at an arbitrary agent node. For a non-root agent $i$ in $F$, let item $j$ be its parent. Then, for any $0 \leq \delta \leq \min\{ b_{i j}, \; 1-b_{i j} \}$, there exists a feasible solution, $(\mb{b}^{\delta}, \mb{q}^{\delta})$ such that $ b^{\delta}_{i j} = b_{i j} - \delta$, $q^{\delta}_{j} = q_{j} - \delta$, $ q^{\delta}_{j'} \geq q_{j'} $ for all $j' \in \cG \backslash \{j\}$, and
    \begin{equation*}
         b^{\delta}_{i'j'} \begin{cases}
            \leq \min\{1, \, 2b_{i'j'} \}& \text{ if } i', j' \in T(i)\\
            =b_{i'j'} & \text{ otherwise}\,,
        \end{cases}
    \end{equation*}
    where $T(x)$ denotes the sub-tree rooted at $x$ in $F$.
\end{restatable}
\begin{proof}[Proof of Lemma \ref{lem:pruned-sol}]
We will iteratively build the solution $(\mb{b}^{\mathrm{pruned}}, \mb{q}^{\mathrm{pruned}})$ satisfying these properties while ensuring it remains feasible. For a vertex $x \in \cA \cup \cG$, let $\pr(x)$ denote its parent in $G_{\mathrm{supp}}(\mb{b}^\star)$, let $C(x)$ denote the set of its children in $G_{\mathrm{supp}}(\mb{b}^\star)$, and let $T(x)$ denote the sub-tree rooted at vertex $x$ in $G_{\mathrm{supp}}(\mb{b}^\star)$.

Consider an item $j$ with $q_{j}^\star < 1/2$. To make the vertex corresponding to $j$ a leaf, the algorithm removes all the edges between item $j$ and its children $C(j)$. To reflect this change, we will update the solution $(\mb{b}^\star, \mb{q}^\star)$ to an intermediate solution $(\widetilde{\mb{b}}, \widetilde{\mb{q}})$ such that the support of $\widetilde{\mb{b}}$ does not contain any edges between item $j$ and its children.
 To maintain feasibility, we require:
        \begin{align}
            \widetilde{q}_{j} &= \widetilde{b}_{\pr(j) j} = b^\star _{\pr(j) j} \notag \\
            \widetilde{b}_{i j} &= 0 \text{ for all } i \in C(j) \label{eq:local-feasibility}
        \end{align}

Note that $q^\star _j < 1/2$ implies $b^\star _{ij} < 1/2$ for each $i \in C(j)$. As a result, the decrease in $b_{ij}$ satisfies
\begin{equation*}
    b^\star _{ij} - \widetilde{b}_{ij} \leq \min\{b^\star _{ij}, 1-b^\star _{ij}\}
\end{equation*}
for each $i \in C(j)$.
So, we update $(\widetilde{\mb{b}}, \widetilde{\mb{q}})$ by iteratively applying Lemma~\ref{lem:redistribution} to edge $(j\rightarrow i)$ with $\delta = b_{ij}$ for each $i \in C(j)$. The updated solution satisfies $\widetilde{b}_{ij} = 0$ for each $i \in C(j)$ and $\widetilde{q}_j = q_j - \sum_{i\in C(j)} b_{ij} = b_{\mathrm{par}(j) j } < 1/2$.
Note that $T(j)$ is the disjoint union of the sub-trees rooted at nodes in $C(j)$. So for distinct $i_1, i_2 \in C(j)$, updating the edge $(j\rightarrow i_1)$ (and the sub-tree for $i_1$) does not affect the $b$ values for any edge in $T(i_2)$ and vice versa. Therefore, by Lemma \ref{lem:redistribution}, we have $q^\star _{j'} \leq \widetilde{q}_{j'}$ for any item $ j' \in T(j)$ and $\widetilde{b}_{i'j'} \leq \min\{1, 2b^\star_{i'j'}\}$ for any $i', j' \in T(j)$.

Since every item with $q_j^\star < 1/2$ must become a leaf, we repeat the above process for any such item. The following fact is crucial to bound the values after multiple pruning processes: Pruning item $j$ only changes $b$ values for edges in $T(j)$, and item $j$ becomes a leaf after that. So, if we prune ancestors of $j$ after pruning $j$, the $b$ values of edges in $T(j)$ do not change further.

Let $(\mb{b}^{\mathrm{pruned}}, \mb{q}^{\mathrm{pruned}})$ be the solution obtained by pruning the set of items $J = \{ j\in \cG: q_j^\star < 1/2\}$ in decreasing order of their height\footnote{Note that pruning items in decreasing order of their height is only an artifact of the analysis. 
The algorithm can prune items with $q_j^\star < 1/2$ in any order.}. Pruning item $j$ does not decrease the $q$ value of any item other than $j$. Therefore, if $q^{\mathrm{pruned}}_j < 1/2$, then $q^{\star}_j < 1/2$, so item $j$ has been pruned and is a leaf. For any item $j$ with $q_j^\star \geq 1/2$, its $q$ value only increases when its nearest ancestor is pruned, and this is the only time its $q$-value changes. So we conclude that $q_j^{\mathrm{pruned}} \geq q_j^\star$ for each $j \in \cG$.  

To establish the third claim of the lemma, observe that the $b$-value of any edge in $G_{\mathrm{supp}}(\mb{b}^\star)$ changes at most twice during the pruning process: If $q^\star_j \geq 1/2$, then item $j$ itself is not pruned, and the $b$ values of edges incident to $j$ may change only when the nearest ancestor of $j$ is pruned. By Lemma~\ref{lem:redistribution}, $b^{\mathrm{pruned}}_{ij}  \leq \min\{1,\; 2b^\star _{ij}\} $ for each $i \in \cA$. If $q^\star_j < 1/2$, the $b$ value of any edge from $j$ to its children becomes zero when $j$ is pruned, satisfying the claim. The $b$ value of the edge $(\pr(j) \rightarrow j)$ does not change when we prune $j$, and it may increase when the nearest ancestor of $j$ in $J$ is pruned. If so, we have $b^{\mathrm{pruned}}_{\pr(j)j}  \leq \min\{1,\; 2b^\star _{\pr(j)j}\}$.
\end{proof}

\subsection{Fractional Matching and Analysis} \label{sec:matching}

In this section, we prove Lemma~\ref{lem:fractional}, which completes the proof of Theorem \ref{thm:non-convex-main}.  

We establish Lemma~\ref{lem:fractional} by proving two inequalities (in Lemmas~\ref{lem:frac-1} and~\ref{lem:frac-2}) about the properties  of $f_{\ncvx}$ at any feasible point whose support is a forest. Lemma~\ref{lem:frac-1} shows that $f_{\ncvx}$ can be upper bounded by a linear function in $\mb{b}$ while only losing a constant factor. 
\begin{lemma}\label{lem:frac-1}
    Let $(\mb{b}, \mb{q})$ be an acyclic solution in $\mathcal{P}(\cA, \cG)$ such that every item with $q_j < 1/2$ is a leaf in $G_{\mathrm{support}}(\mb{b})$. Let $S: \cA \rightarrow 2^{\cG}$ be a function such that for each agent $i$, $S(i)$ is a subset of leaf items connected to agent $i$ in $G_{\mathrm{supp}}(\mb{b})$, and  $S(i)$ contains all children of agent $i$ with $q_j < 1/2$. 
 Then 
\begin{align*}
 \sum_{i\in \cA} w_i \left(\sum_{j \notin S(i)} b_{ij} \log{ v_{ij}} +  \sum_{j \in S(i)} b_{ij} \log\left(\sum_{j \in S(i)}v_{ij}\right)\right)  &\geq   f_{\ncvx}(\mb{b},\mb{q}) -\log 2 -\frac{1}{2e}. 
\end{align*}
\end{lemma}

Lemma~\ref{lem:frac-2} demonstrates how the linear function obtained in Lemma~\ref{lem:frac-1} can be used as a lower bound for the  maximum weight matching with the augmented weight function. A crucial component of the proof of this lemma is the fact that any feasible $\mb{b}$ in $\mathcal{P}(\cA, \cG)$ corresponds to a point in the matching polytope where all agents are matched.

\begin{lemma}\label{lem:frac-2}
    Let $(\mb{b}, \mb{q})$ be an acyclic solution in $\mathcal{P}(\cA, \cG)$ such that every item with $q_j < 1/2$ is a leaf in $G_{\mathrm{support}}(\mb{b})$. Let $S: \cA \rightarrow 2^{\cG}$ be a function such that for each agent $i$, $S(i)$ is a subset of leaf items connected to agent $i$ in $G_{\mathrm{supp}}(\mb{b})$, and  $S(i)$ contains all children of agent $i$ with $q_j < 1/2$. 
 Then, there exists a matching $M$ in $G_{\mathrm{supp}}(\mb{b})$ between vertices in $\cA$ and $\{\cG\backslash \cup_i\{S(i)\} \}$  such that
\begin{align}
 \sum_{i \in \cA} w_i \log\left(v_{iM(i)} + \sum_{j \in S(i)} v_{ij}\right) \geq  \sum_{i\in \cA} w_i \left(\sum_{j \notin S(i)} b_{ij} \log{ v_{ij}} +  \sum_{j \in S(i)} b_{ij} \log\left(\sum_{j \in S(i)}v_{ij}\right)\right)\,, \label{eq:eq1}
\end{align}
where $v_{iM(i)} = 0$ if agent $i$ is not matched in $M$.
\end{lemma}

Lemma~\ref{lem:frac-1} and Lemma~\ref{lem:frac-2} together establish Lemma~\ref{lem:fractional}. In the rest of this section, we provide the proofs of Lemma~\ref{lem:frac-1} and Lemma~\ref{lem:frac-2}.

\begin{proof}[Proof of Lemma~\ref{lem:frac-1}]
Let $S := \cup_i \{S(i)\}$. Recall that 
 \begin{align}
         f_{\ncvx}(\mb{b}, \mb{q}) &= \sum_{i\in \cA} w_i \sum_{j \in \cG} b_{ij} \log{v_{ij}} - \sum_{i\in \cA} w_i \sum_{j \in \cG } b_{ij} \log{q_j} \notag\\
         &= \sum_{i\in \cA} w_i \sum_{j \notin S(i)} b_{ij} \log{v_{ij}} - \sum_{i\in \cA} w_i \sum_{j \notin S(i)} b_{ij} \log{q_j} + \sum_{i\in \cA} w_i \left(\sum_{j \in S(i)} b_{ij} \log{v_{ij} } - b_{ij} \log{b_{ij}}  \right),\label{eq:e6}
    \end{align}
where the last equation follows from the fact that every item in $S(i)$ is a leaf, i.e., if $j \in S(i)$, then $b_{i'j} = 0$ for every $i' \neq i$. 

For an item $j \notin S$, we have $q_j \geq  1/2$. As a result, \begin{align}
     -\sum_{i\in \cA} w_i \, b_{ij} \log{q_j}  \leq  \log{2}\,\sum_{i\in \cA} w_i \, b_{ij}.  \label{eq:e4}
\end{align}

Plugging this bound into equation~\eqref{eq:e6} gives
\begin{align}
    f_{\ncvx}(\mb{b}, \mb{q}) \leq & \sum_{i\in \cA} w_i \sum_{j \notin S(i)} b_{ij} \log{v_{ij}} -  \sum_{i\in \cA}  w_i \sum_{j \notin S(i)} b_{ij} \log{2} + \sum_{i\in \cA} w_i \left(\sum_{j \in S(i)} b_{ij} \log{v_{ij} } - b_{ij} \log{b_{ij}} \right).\label{eq:e7}
    \end{align}

As $\mb{b} \in \mathcal{P}(\cA, \cG)$, we have $ \sum_{j \notin S(i)} b_{ij} = 1- \sum_{j \in S(i)} b_{ij} $
for every agent $i$. Substituting this in equation~\eqref{eq:e7} yields
\begin{align}
         f_{\ncvx}(\mb{b}, \mb{q}) 
         &\leq  \sum_{i\in \cA} w_i \sum_{j \notin S(i)} b_{ij} \log{v_{ij}}  + \sum_{i\in \cA} w_i \log{2}+ \sum_{i\in \cA} w_i \left(\sum_{j \in S(i)} b_{ij} \log{v_{ij} } - b_{ij} \log{b_{ij}} - b_{ij} \log{2} \right) \notag\\
         &= \sum_{i\in \cA} w_i \sum_{j \notin S(i)} b_{ij} \log{v_{ij}}  +  \log{2}+ \sum_{i\in \cA} w_i \left(\sum_{j \in S(i)} b_{ij} \log{v_{ij} } - b_{ij} \log{b_{ij}} - b_{ij} \log{2} \right), \label{eq:ep8}
\end{align}
where the last equation follows from $\sum_i w_i = 1$.

For each agent $i \in \cA$, Claim~\ref{cor:val} implies that
\begin{align*}
    \sum_{j \in S(i)} b_{ij} \log{v_{ij}} - b_{ij} \log{b_{ij}}  \leq \sum_{j\in S(i)}b_{ij}\log\left(\sum_{j\in S(i)} v_{ij}\right) - \sum_{j\in S(i)} b_{ij}\log\left(\sum_{j\in S(i)} b_{ij}\right) .
\end{align*}

So, for any agent $i$, \begin{align}
    \sum_{j \in S(i)} &b_{ij} \log{v_{ij}} - b_{ij} \log{b_{ij}} - b_{ij} \log{2} \notag\\
    &\leq \sum_{j\in S(i)} b_{ij}\log\left(\sum_{j\in S(i)} v_{ij}\right) - \sum_{j\in S(i)} b_{ij}\log\left(\sum_{j\in S(i)} b_{ij}\right)- \sum_{j\in S(i)} b_{ij}\log{2}  \notag\\
    &\leq \sum_{j\in S(i)} b_{ij}\log\left(\sum_{j\in S(i)} v_{ij}\right) + \frac{1}{2e}\,, \label{eq:ep9}
\end{align}
where the last inequality follows from  $-x\log(x) -x\log{2}\leq 1/(2e)$ for all $x \geq 0$ applied to $x = \sum_{j\in S(i)} b_{ij}$.

Substituting equation \eqref{eq:ep9} in equation~\eqref{eq:ep8}, we get
\begin{align*}
    f_{\ncvx}(\mb{b}, \mb{q}) &\leq  \sum_{i\in \cA} w_i \sum_{j \notin S(i)} b_{ij} \log{v_{ij}} + \log{2} + \sum_{i\in \cA}w_i \left(\sum_{j \in S(i)} b_{ij} \log\left(\sum_{j \in S(i)}v_{ij}\right) + \frac{1}{2e} \right) \\
    &= \sum_{i\in \cA} w_i \left(\sum_{j \notin S(i)} b_{ij} \log{ v_{ij}} +  \sum_{j \in S(i)} b_{ij} \log\left(\sum_{j \in S(i)}v_{ij}\right) \right) +\log{2} + \frac{1}{2e}\,,
\end{align*}
where the last inequality again follows from $\sum_{i\in \cA} w_i = 1$.
\end{proof}

\begin{proof}[Proof of Lemma~\ref{lem:frac-2}] In this proof, we will analyze a matching that either assigns the bundle $S(i)$ to an agent or a single item $j\notin \cup_i \{S(i) \}$. Observe that the algorithm clearly finds an assignment with a larger objective as
 \begin{equation*}
      \log\left(v_{iM(i)} + \sum_{j \in S(i)} v_{ij}  \right) \geq \max\left\{\log v_{iM(i)}, \,\log\left(\sum_{j \in S(i)} v_{ij}\right) \right\}.
 \end{equation*}

So, for each agent $i \in \cA$, we create a new leaf item $\ell_i$ with $v_{i \ell_i} = \sum_{j \in S(i)} v_{ij}$ corresponding to the set of items in $S(i)$. Define $S := \cup_i \{S(i)\}$ and $\widetilde{\cG} := \{\cG \backslash S\} \cup \{\ell_i\}_{i \in \cA}$. We show that the maximum weight matching in the bipartite graph $(\cA, \widetilde{\cG})$ suffices to prove the lemma. As the matching polytope is integral, it is enough to demonstrate the existence of a fractional matching of a large value. 

 Using $\mb{b}$, we define fractional assignment variables $\mb{x}$ as follows:
\begin{align*}
    x_{ij} &:= b_{ij} \quad \forall i \in \cA, j \in \{\cG\backslash L\} \\
    x_{i \ell_i} &:= \sum_{j \in S(i)} b_{ij} \quad \forall i \in \cA.
\end{align*}
The L.H.S. of equation~\eqref{eq:eq1} can be stated in terms of $\mb{x}$ as
\begin{equation}
\sum_{i\in \cA} w_i \left(\sum_{j \notin S(i)} b_{ij} \log{ v_{ij}} + \sum_{j \in S(i)} b_{ij} \log \left(\sum_{j \in S(i)}v_{ij}\right)\right) =    \sum_{i\in \cA} \sum_{j \in \widetilde{\cG}}x_{ij} \, w_{i}\log{ v_{ij} }. \label{eq:x1}
\end{equation}
Observe that $\mb{x}$ lies in the convex hull of matchings between agents $\cA$ and items $\widetilde{\cG}$ in which every agent is matched as $\mb{x}$ satisfies the following properties:
\begin{align*}
    \sum_{j\in \widetilde{\cG}} x_{ij}  = \sum_{j\notin S(i)} b_{ij} +  \sum_{j\in S(i)} b_{ij} &= 1 \quad \forall i \in \cA\\
   \sum_{i\in \cA} x_{ij} &\leq 1 \quad \forall j \in \widetilde{\cG}.
\end{align*}
Here, for item $j\notin S$, the second inequality is inherited from the feasibility of $\mb{b}$. The constraint for $\ell_{i'}$ for some $i'\in \cA$ is implied by the constraint $\sum_{i\in \cA} x_{ij}= x_{i'j}=\sum_{j \in S(i)} b_{ij} \leq\sum_{j \in \cG} b_{ij}  \leq 1$, where the last constraint again follows from the feasibility of $\mb{b}$.

Using the integrality of the matching polytope, there exists a matching $\widetilde{M}: \cA\rightarrow \widetilde{\cG}$ such that \begin{equation}
\sum_{i\in \cA} \sum_{j \in \widetilde{\cG}}x_{ij} \, w_{i}\log{ v_{ij} } \leq  \sum_{i\in \cA}w_i\, \log{v_{i\widetilde{M}(i)}} . \label{eq:x2}
\end{equation}

Now consider a matching $M: \cA \rightarrow \cG$ with $M(i) = \emptyset$ if $\widetilde{M}(i) = \ell_i$, and $M(i) = \widetilde{M}(i)$ otherwise. Then
\begin{equation}
     \sum_{i\in \cA}w_i\, \log{v_{i\widetilde{M}(i)}}  \leq  \sum_{i\in \cA}w_i\, \log\left(v_{i M(i)} + \sum_{j \in S(i)} v_{ij} \right).\label{eq:x3}
\end{equation}
Then equations~\eqref{eq:x1},~\eqref{eq:x2}, and~\eqref{eq:x3} together imply  \allowdisplaybreaks
\begin{equation*}
 \sum_{i\in \cA}w_i\, \log\left(v_{i M(i)} + \sum_{j \in S(i)} v_{ij} \right) \geq \sum_{i\in \cA} w_i \left(\sum_{j \notin S(i)} b_{ij} \log{ v_{ij}} + \sum_{j \in S(i)} b_{ij} \log \left(\sum_{j \in S(i)}v_{ij}\right)\right).
\end{equation*}\end{proof}

\section{Conclusion and Open Questions}
This paper introduces a convex and a non-convex relaxation for the weighted (asymmetric) Nash Social Welfare problem to give an $O(\exp\left(2\kl{\mb{w}}{\mb{u}}\right))$-approximation. Both of these relaxations play a crucial role in obtaining the approximation algorithm for the problem. There are two natural open questions. First, is the factor $\exp\left(2\kl{\mb{w}}{\mb{u}}\right)$ necessary in the approximation guarantee? Equivalently, is it possible to obtain a constant factor approximation for the weighted Nash Social Welfare problem? It is important to emphasize that we lose the $\exp\left(2\kl{\mb{w}}{\mb{u}}\right)$ when relating the objectives of the two relaxations; we only lose a constant factor when rounding the non-convex relaxation. A direct approach may exist to solve the non-convex formulation that gives an improved approximation guarantee.

The second question is whether the techniques introduced in this work generalize to more general valuation functions, particularly submodular valuations for the weighted Nash Social Welfare problem. While there are constant factor approximation algorithms for unweighted Nash Social Welfare with submodular valuations, obtaining anything better than $O(nw_{\max})$-approximation for the weighted variant of the problem remains an open question.

\newpage
\bibliographystyle{alpha}
\bibliography{references}

%%%%%%%%%%%%%%%%%%%%%%%%%%%%%%%%%%%%%%%%%%%%%%%%%%%%%%%%%%%%
\appendix
\section{Omitted Proofs and Lemmas} \label{sec:proofs}

\begin{proof}[Proof of Lemma~\ref{lem:relaxation}]
Let $\sigma: \cG \rightarrow \cA$ be the optimal assignment of an instance of NSW, $(\cA, \cG, \mb{v}, \mb{w})$. 
For each agent $i \in \cA$, define $V_i = \sum_{j\in \sigma^{-1}(i)} v_{ij}$. Using $\sigma$, we define a vector $\mb{b} \in \mathcal{P}(\cA, \cG)$ as
\begin{equation*}
     b_{ij} := \begin{cases}
        \frac{v_{ij}}{V_i} & \text{ if } \sigma(j) = i \\
        0 & \text{ otherwise}.
    \end{cases}
\end{equation*}

It is easy to verify that $\sum_{i \in \cA} b_{ij} \leq 1$ for each $j \in \cG$ and $\sum_{i \in \cG} b_{ij} = 1$ for each $i \in \cA$.
We will now show that $f_{\cvx}(\mb{b})$ and $f_{\ncvx}(\mb{b})$ are both equal to $\nsw(\sigma)$.
\begin{align*}
    f_{\cvx}(\mb{b}) &= \sum_{j\in \cG}\frac{w_{\sigma(j)}v_{\sigma(j)j}}{V_{\sigma(j)}} \log{v_{\sigma(j)j}} - \sum_{j \in \cG}  \frac{w_{\sigma(j)}v_{\sigma(j)j}}{V_{\sigma(j)}} \log\left( \frac{w_{\sigma(j)}v_{\sigma(j)j}}{V_{\sigma(j)}} \right) + \sum_{i \in \cA} w_i \log{w_i} \\
    &= \sum_{j \in \cG}  \frac{w_{\sigma(j)}v_{\sigma(j)j}}{V_{\sigma(j)}}\log \left(\frac{V_{\sigma(j)}}{w_{\sigma(j)}} \right) + \sum_{i \in \cA} w_i \log{w_i} \\
    &= \sum_{i \in \cA} w_i \sum_{j \in \sigma^{-1}(i)} \frac{v_{ij}}{V_i} \log \left(\frac{V_{i}}{w_{i}}\right) + \sum_{i \in \cA} w_i \log{w_i} \\
    &\stackrel{(i)}{=} \sum_{i \in \cA} w_i \log \left(\frac{V_{i}}{w_{i}}\right) + \sum_{i \in \cA} w_i \log{w_i} 
    =  \sum_{i \in \cA} w_i \log{V_i} = \nsw(\sigma)\,,
    \end{align*}
    where $(i)$ follows from definition of $V_i$. 

    Similarly, we have
    \begin{align*}
    f_{\ncvx}(\mb{b}) &= \sum_{j\in \cG}\frac{w_{\sigma(j)}v_{\sigma(j)j}}{V_{\sigma(j)}} \log{v_{\sigma(j)j}} - \sum_{j \in \cG}  \frac{w_{\sigma(j)}v_{\sigma(j)j}}{V_{\sigma(j)}} \log\left( \frac{v_{\sigma(j)j}}{V_{\sigma(j)}} \right) \\
    &= \sum_{j \in \cG}  \frac{w_{\sigma(j)}v_{\sigma(j)j}}{V_{\sigma(j)}}\log{V_{\sigma(j)}} = \sum_{i \in \cA} w_i \sum_{j \in \sigma^{-1}(i)} \frac{v_{ij}}{V_i} \log {V_{i}}\\
    &=  \sum_{i \in \cA} w_i \log{V_i} = \nsw(\sigma).
\end{align*}

For the second claim in the lemma, when $w_i = 1/n$ for each $i$, for any $\mb{b} \in \mathcal{P}(\cA, \cG)$, we have
\begin{align*}
    f_{\cvx}(\mb{b}) &= \frac{1}{n}\sum_{i\in \cA} \sum_{j\in \cG} b_{ij}\log{v_{ij}} - \frac{1}{n}\sum_{i\in \cA} \sum_{j\in \cG} b_{ij}\log\left(\frac{\sum_{i\in \cA} b_{ij}}{n}\right) - \log{n} \\
    &= \frac{1}{n}\sum_{i\in \cA} \sum_{j\in \cG} b_{ij}\log{v_{ij}} - \frac{1}{n}\sum_{i\in \cA} \sum_{j\in \cG} b_{ij}\log\left(\sum_{i\in \cA} b_{ij}\right) +\frac{1}{n}\sum_{i\in \cA} \sum_{j\in \cG} b_{ij}\log{n} - \log{n}  \\
    &= \frac{1}{n}\sum_{i\in \cA} \sum_{j\in \cG} b_{ij}\log{v_{ij}} - \frac{1}{n}\sum_{i\in \cA} \sum_{j\in \cG} b_{ij}\log\left(\sum_{i\in \cA} b_{ij}\right),
\end{align*}
where we used $\sum_{j \in \cG} b_{ij} = 1$ for every $i$ in the last inequality.
Similarly, substituting $w_i = 1/n$ for each $i$ in $f_{\ncvx}$ completes the proof.
\end{proof}

\begin{proof}[Proof of Lemma~\ref{lem:diff}]
We will show that 
\begin{equation*}
    f_{\cvx}(\mb{b})-f_{\ncvx}(\mb{b})= \kl{\mb{w}}{\mb{u}}- \kl{\mu}{\theta}\,,
\end{equation*}

where $\mu, \theta$ are two probability distributions on $\cG$ given by
\begin{equation*}
   \mu(j) = \sum_{i \in \cA} w_i \, b_{ij} \quad \text{and} \quad \theta(j) = \frac{\sum_{i \in \cA} b_{ij}}{n}. 
\end{equation*}

Using $\sum_{i\in \cA} w_i = 1$ and $\sum_{j \in \cG} b_{ij} = 1$ for each $i \in \cA$, one can verify that $\sum_{j \in \cG} \mu(j) = 1 = \sum_{j \in \cG} \theta(j)$.

Expanding the difference between the functions gives
\begin{align*}
    f_{\cvx}(\mb{b})-f_{\ncvx}(\mb{b}) &= \sum_{i \in \cA} w_i \log{w_i}-\sum_{j\in \cG} \sum_{i\in \cA} w_i \, b_{ij} \log\left(\sum_{i\in \cA} w_i\, b_{ij}\right) +  + \sum_{i,j}w_i \, b_{ij} \log\left(\sum_{i\in \cA}  b_{ij}\right) \\
    &= \sum_{i \in \cA} w_i \log{w_i} -\sum_{j\in \cG} \sum_{i\in \cA} w_i \, b_{ij} \log\left(\frac{\sum_{i \in \cA} w_i\, b_{ij}}{\sum_{i\in \cA}  b_{ij}}\right) \\
    &= \sum_{i \in \cA} w_i \log{w_i} + \sum_{j\in \cG} \sum_{i\in \cA} w_i \, b_{ij} \log{n} -\sum_{j \in \cG} \mu(j) \log\left(\frac{\mu(j)}{\theta(j)}\right) \\
    &= \sum_{i \in \cA} w_i \log(n w_i)- \sum_{j\in \cG}\mu(j) \log\left(\frac{\mu(j)}{\theta(j)}\right)  \tag{using $\sum_{j} b_{ij} = 1$}\\
    &=\kl{\mb{w}}{\mb{u}} - \kl{\mu}{\theta}.
\end{align*}

As $ D_{KL}(\mu, \theta) \geq 0$, the above equation implies 
\begin{equation*}
    f_{\cvx}(\mb{b})-f_{\ncvx}(\mb{b}) \leq \kl{\mb{w}}{\mb{u}} .
\end{equation*}

For the lower bound, it suffices to show that $\kl{\mu}{\theta} \leq \kl{\mb{w}}{\mb{u}}$. To see this, we expand the definition:
\begin{align*}
    \kl{\mu}{\theta}
    &= \sum_{j\in\cG} \left(\sum_{i\in \cA} w_ib_{ij}\right) \log\left(\frac{n \sum_{i\in \cA} w_ib_{ij}}{\sum_{i\in\cA} b_{ij}} \right) \\
    &= \log(n) + \sum_{j\in\cG} \left(\sum_{i\in \cA} w_ib_{ij}\right) \log\left(\frac{\sum_{i\in \cA} w_ib_{ij}}{\sum_{i\in\cA} b_{ij}} \right) \\
    &= \log(n) + \sum_{j \in \cG} \left(\sum_{i \in A} b_{ij}\right) \left(\sum_{i \in \cA} \frac{b_{ij}}{\sum_{i \in A} b_{ij}} \cdot w_i \right) \log \left(\sum_{i \in \cA} \frac{b_{ij}}{\sum_{i \in A} b_{ij}} \cdot w_i \right)\\
    &\leq \log(n) + \sum_{j \in \cG} \left(\sum_{i \in A} b_{ij}\right) \left(\sum_{i \in \cA} \frac{b_{ij}}{\sum_{i \in A} b_{ij}} \right) w_i \log(w_i) \\
     &= \log(n) + \sum_{j \in \cA}  w_i \log(w_i) \sum_{j\in\cG} b_{ij} \\
    &= \log(n) + \sum_{i\in \cA} w_i\log(w_i) = \kl{\mb{w}}{\mb{u}}.
\end{align*}
Here, the only inequality uses the convexity of $x\log(x)$, and the last equality follows from the feasibility of $\mb{b}$.
\end{proof}

\begin{proof}[Proof of Lemma \ref{lem:redistribution}]
For $x \in \cA \cup \cG$, let $C(x)$ denote the children of node $x$ in $F$ and let $T(x)$ denote the sub-tree rooted at node $x$. We will prove this lemma by induction on the height of agent $i$, building $(\mb{b}^{\delta}, \mb{q}^{\delta}) \in \mathcal{P}(\cA, \cG)$ in the process.

  For the base case, assume agent $i$ has height $1$, i.e., $T(i)$ consists of only leaf item nodes that are the children of node $i$.
  We define a new vector $\mb{b}^\delta$ with $b^\delta_{i'j'} = b_{i'j'}$ for any $i' \neq i$ and $j' \in \cG$.
  Note that setting $b^{\delta}_{i j} = b_{i j} - \delta$ and $q^{\delta}_{j} = q_{j} - \delta$ only violates the Agent constraint for agent $i$. So we will update the values of $\mb{b}$ in $T(i)$ to make the solution feasible.

  By the feasibility of $\mb{b}$, $b_{i j} + \sum_{k\in C(i)} b_{i k} = 1$, and for every item node $k \in C(i)$, $q_{k} = b_{i k} < 1$. Using Lemma~\ref{lem:rebalance} with $\alpha = b_{ij}$ and $\beta_{k} = b_{ik}$, there exist $\delta_k$ for each $k \in C(i)$ such that
    \begin{align*}
         b_{ij} - \delta + \sum_{k \in C(i)} b_{i k}(1+\delta_k) &= 1 \\
        b_{i k} (1+\delta_k) &\leq 1 \quad \forall k \in  C(i)\\
       0\leq \delta_k &\leq 1 \quad \forall k \in  C(i).
    \end{align*}

So, for each $k \in C(i)$, we set $ b^{\delta}_{i k} = b_{i k}(1 + \delta_k)$. Note that $b^{\delta}_{i k} \leq 1$, and as $\delta_k \leq 1$, we have
\begin{equation*}
    b^{\delta}_{i k} = b_{i k}(1 + \delta_k) \leq 2b_{ik}.
\end{equation*}
As every item in $C(i)$ is a leaf, we also have
\begin{equation*}
    q_k \leq q^{\delta}_k = b^{\delta}_{i k} = b_{i k} (1+\delta_k)\leq 1
\end{equation*}
for each item $k \in C(i)$. The Agent constraint for agent $i$ satisfies
\begin{equation*}
   \sum_{k \in  C(i)} b^{\delta}_{i k} =  b_{i j} -\delta + \sum_{k \in  C(i)} b_{i k} (1+\delta_k) = 1. 
\end{equation*}
     
Therefore, $\mb{b}^{\delta} \in \mathcal{P}(\cA, \cG)$ and $b^{\delta}_{i'j'} \leq \min\{1, 2b_{ij'}\}$ for each $ j' \in T(i)$.

For the induction hypothesis, assume that the lemma is true whenever the height of agent $i$ is at most $\ell-1$ for some integer $\ell > 1$. 
We now show that the statement also holds when the height of agent $i$ is $\ell$.

Again, setting $b^{\delta}_{i j} = b_{i j} - \delta$ and $q^{\delta}_{j} = q_{j} - \delta$ violates the Agent constraint for agent $i$. Similar to the base case, we can find $\delta_k \in (0,1)$ for each $k \in C(i)$ such that $b_{i k} (1+\delta_k) \leq 1$ and
 \begin{equation*}
      b_{i j} - \delta + \sum_{k \in C(i)} b_{ik}(1+\delta_k) = 1.
 \end{equation*}

Setting $b^{\delta}_{i k} = b_{i k} (1+\delta_k)$ for each $k\in C(i)$ will ensure that $\mb{b}^{\delta}$ satisfies the Agent constraint for agent $i$. However, this can violate the Item constraint for some item $k \in C(i)$, as $q^{\delta}_k = q_k + \delta_{k} b_{ik}$.  So, we inductively update the values of $\mb{b}^\delta$ and $\mb{q}^\delta$ for the sub-tree rooted at item $k$ for which such a violation occurs.

Consider an item $k \in C(i)$ such that  $q^{\delta}_k = q_k + \delta_{k} b_{ik} > 1$. So we decrease $b_{i'k}$ for each $i' \in C(k)$ to ensure that $q_k^{\delta}$ is at most $1$ as follows. Define $\gamma := q_k + \delta_{k} b_{ik} - 1$. 
Using the fact that $q_k = \sum_{i' \in C(k)} b_{i'k} + b_{ik} $, we bound $\gamma$ as follows.
\begin{align*}
    \gamma &= q_k + \delta_k b_{ik} - 1 = \sum_{i' \in C(k)} b_{i'k} + b_{ik} + \delta_k b_{ik} - 1 \\
    &\leq \sum_{i' \in C(k)} b_{i'k}. \tag{using $b_{ik}(1+\delta_k) \leq 1$}
\end{align*}
 Therefore, there exist numbers $\gamma_{i'}\geq 0$ for each $i' \in C(k)$ such that $\gamma_{i'} \leq b_{i'k}$ and
 $\sum_{i'\in C(k)} \gamma_{i'} =  \gamma$.

We would like to update $b^{\delta}_{i'k} = b_{i'k} - \gamma_{i'}$ for each $i' \in C(k)$, but this violates the Agent constraint for agent $i'$ when $\gamma_{i'} > 0$. We inductively update the solution for subtree $T(i')$ as follows.

First, note that $q^{\delta}_k = 1 \geq q_k$ after this update, as shown below.
\begin{align*}
   q^{\delta}_k &= b^{\delta}_{i k} + \sum_{i' \in C(k)} b^{\delta}_{i' k} = b_{i k}(1+\delta_{k}) + \sum_{i'\in C(k)} (b_{i'k}-\gamma_{i'})  = q_k + \delta_i b_{ik} - \sum_{i'\in C(k)} \gamma_{i'} \\
   &= 1 - \gamma + \sum_{i'\in C(k)} \gamma_{i'} \tag*{(by definition of $\gamma$)} \\ 
   &= 1 . \tag*{(as $\sum_{i'\in C(k)} \gamma_{i'} = \gamma$)}
\end{align*}
So now, $(\mb{b}^\delta, \mb{q}^\delta)$ only violates Agent constraints for agents in $C(k)$.

We claim that for each agent $i' \in C(k)$
\begin{equation}
      \gamma_{i'} \leq \min\{b_{i' k}, \; 1-b_{i' k}\}. \label{eq:gam}
\end{equation}
Before proving this inequality, we use it to complete the proof.

 Using the induction hypothesis, for each $i' \in C(k)$, there exists feasible $(\mb{b}^{\gamma_{i'}}, \mb{q}^{\gamma_{i'}})$ which differs from $(\mb{b}^\delta,\mb{q}^\delta)$ only in the sub-tree rooted at $i'$ such that for any $\hat{j} \in T(i')$, 
 \begin{equation*}
    q^{\gamma_{i'}}_{\hat{j}} \geq q^{\delta}_{\hat{j}} = q_{\hat{j}} .
\end{equation*} 
 and for any $\hat{i}, \hat{j} \in T(i')$, 
\begin{equation*}
    b^{\gamma_{i'}}_{\hat{i}\hat{j}} \leq \min\{1, 2 \cdot b^{\delta}_{\hat{i}, \hat{j}}\} = \min\{1, 2\cdot b_{\hat{i}, \hat{j}}\}.
\end{equation*} 

So for each $i' \in C(k)$ with $\gamma_{i'} > 0$, we set 
$b^{\delta}_{\hat{i} \hat{j}} = b^{\gamma_{i'}}_{\hat{i}\hat{j}}$ for every $\hat{i}, \hat{j} \in T(i')$ to get the required solution.

We now only need to establish equation~\eqref{eq:gam}. By definition, $\gamma_{i'} \leq b_{i'k}$ for each $i' \in C(k)$.  Additionally, $\gamma_{i'} \leq \gamma$, so it suffices to show that $\gamma \leq 1-b_{i'k}$ for every $i' \in C(k)$. Recall that
\begin{align*}
    \gamma &= q_k + \delta_k b_{ik} - 1 \\
    &\stackrel{(i)}{\leq} \delta_k b_{ik} \stackrel{(ii)}{\leq} b_{ik} \stackrel{(iii)}{\leq}  q_k - b_{i' k} \stackrel{(iv)}{\leq} b_{i'k}.
\end{align*}
Here, $(i)$ and $(iv)$ follow from $q_k \leq 1$, $(ii)$ follows from $\delta_k \leq 1$, and $(iii)$ holds as $b_{ik} + \sum_{i'\in C(k)} b_{i'k} = q_k$.
This completes the proof of \eqref{eq:gam}. 
\end{proof}

\begin{lemma} \label{lem:rebalance}
Let $\alpha > 0$ and $\beta_1, \ldots, \beta_k > 0$ with $\alpha + \sum_{j=1}^k \beta_i = 1$. For any $0 < \delta \leq \min \{\alpha,\, 1-\alpha\} $, there exist real numbers $\delta_1, \ldots, \delta_k$ such that
    \begin{align}
         \alpha - \delta + \sum_{j\in [k]} \beta_j(1+\delta_j)  &= 1 \label{eq:syst0}\\
        \beta_{j} (1+\delta_j) &\leq 1 \quad \forall j \in [k] \notag\\
       0 \leq \delta_j &\leq 1  \quad \forall j \in [k]. \notag
    \end{align}
\end{lemma}

\begin{proof}
    As the above system contains only linear constraints in $\delta$, we use Farkas' Lemma to show the existence of $\{\delta_j\}_{j=1}^k$. Re-arranging the constraints gives
    \begin{align}
         \sum_{j\in [k]} \beta_j \delta_j  &= \delta \label{eq:sys1}\\
        \beta_{j} \delta_j  &\leq 1 - \beta_j \quad \forall j \in [k] \notag\\
      0 \leq \delta_j &\leq 1  \quad \forall j \in [k] \notag
    \end{align}

If there do not exist real numbers $\{\delta_j\}_{j=1}^k$ satisfying~\eqref{eq:sys1}, then by Farkas' Lemma, there exist real numbers $\eta, \{\gamma_j\}_{j=1}^k, \{\lambda_j\}_{j=1}^k$ such that \begin{align}
    \beta_j\eta + \beta_j \gamma_j + \lambda_j &\geq 0 \quad \forall j \in [k] \label{eq:e1}\\
    \gamma_j, \lambda_j &\geq 0 \notag\\
    \delta \eta + \sum_{j \in [k]} (1-\beta_j) \gamma_j + \sum_{j\in [k]}  \lambda_j &< 0 \label{eq:e2}
\end{align}
Adding equation~\eqref{eq:e1} for all $j \in [k]$, we get
\begin{equation*}
     \eta\sum_{j\in [k]} \beta_j+ \sum_{j\in [k]} \beta_j \gamma_j + \sum_{j\in [k]} \lambda_j \geq 0.
\end{equation*}
 Since $\alpha + \sum_{j \in [k]} \beta_j = 1$, this implies $\eta(1-\alpha)+ \sum_{j\in [k]} \beta_j \gamma_j + \sum_{j\in [k]} \lambda_j \geq 0 $. In addition, since $\beta_i > 0$, we also have $\alpha < 1$. Therefore, dividing by $1-\alpha$ and re-arranging gives  \begin{equation}
    \sum_{j\in [k]} \frac{\beta_j \gamma_j}{1-\alpha} + \sum_{j\in [k]} \frac{\lambda_j}{1-\alpha} \geq -\eta .\label{eq:e21}
\end{equation}
On the other hand, equation~\eqref{eq:e2} implies
\begin{equation}
    -\eta > \sum_{j\in [k]} \frac{(1-\beta_j)\gamma_j}{\delta} + \sum_{j\in [k]} \frac{\lambda_j}{\delta}. \label{eq:e22}
\end{equation}

On comparing equations~\eqref{eq:e21} and~\eqref{eq:e22}, we obtain
\begin{equation}
    \sum_{j\in [k]} \frac{\beta_j \gamma_j}{1-\alpha} + \sum_{j\in [k]} \frac{\lambda_j}{1-\alpha} > \sum_{j\in [k]} \frac{(1-\beta_j)\gamma_j}{\delta} + \sum_{j\in [k]} \frac{\lambda_j}{\delta}. \label{eq:e9}
\end{equation}
We will now derive a contradiction to~\eqref{eq:e9}.

As $\delta \leq 1-\alpha$ , we have $1/(1-\alpha) \leq 1/\delta$, and therefore,
\begin{align}
   \sum_{j \in [k]} \frac{\lambda_j}{1-\alpha} \leq \sum_{j \in [k]} \frac{\lambda_j}{\delta}\,, \label{eq:e10}
\end{align}
where we use the fact that $\lambda_j > 0$ for all $j \in [k]$.

In addition, for any $j \in [k]$ \begin{align}
    \frac{\beta_j}{1-\alpha} - \frac{(1-\beta_j)}{\delta} \leq \frac{\beta_j}{1-\alpha} - \frac{(1-\beta_j)}{\alpha}
    = \frac{\alpha+\beta_j-1}{\alpha(1-\alpha)} \leq 0. \label{eq:e11}
\end{align}
Here, the first inequality follows from $\delta \leq \alpha$, and the last inequality follows from the facts that $\alpha + \sum_{j\in [k]} \beta_j = 1$ and $\alpha, \beta_j > 0$.

On adding equation~\eqref{eq:e10} and equation~\eqref{eq:e11} for all $j \in [k]$, we obtain
\begin{equation*}
     \sum_{j\in [k]} \frac{\beta_j \gamma_j}{1-\alpha} + \sum_{j\in [k]} \frac{\lambda_j}{1-\alpha} \leq   \sum_{j\in [k]} \frac{(1-\beta_j)\gamma_j}{\delta} + \sum_{j \in [k]} \frac{\lambda_j}{\delta} \, ,
\end{equation*}
which contradicts equation~\eqref{eq:e9}. Therefore, there exist real numbers $\{\delta_j\}_{j=1}^k$ satisfying~\eqref{eq:syst0}.
\end{proof}

\section{Relationships Between the Mathematical Programs}\label{sec:mathematical_programs}
This section provides the proof of Theorem \ref{thm:sym-eq} by establishing a relationship between two natural convex programming relaxations for the unweighted Nash Social Welfare problem. We then build upon this relationship to derive \nameref{eq:b-cvx} for the weighted Nash Social Welfare problem.

To ensure that the optimum values of all the convex programs mentioned below are bounded, we assume that the instance of Nash Social Welfare $(\cA, \cG, \mb{v}, \mb{w})$ satisfies the following assumption.

\begin{assumption} \label{assumption:1}
    Let $G[\cG, \cA, \mb{v}]$ denote the support graph of the valuation function. The support graph is the bipartite graph between agents and items with an edge between agent $i$ and item $j$ iff $v_{ij}>0$. We assume that there exists a matching of size $|\cA|$ in $G[\cG, \cA, \mb{v}]$. In other words, the objective of the Nash Social Welfare problem is not zero for $(\cA, \cG, \mb{v}, \mb{w})$.
\end{assumption}
It is straightforward to verify this assumption given an instance of Nash Social Welfare.

The proof of Theorem \ref{thm:sym-eq} uses the following two results. The first result is the classical Sion's Minimax Theorem, which can be found as Corollary 3.3 from \cite{sion_1958}. 
\begin{theorem}[Sion's Minimax Theorem]
\label{lem:sion-minimax}
    Let $M$ and $N$ be convex spaces, one of which is compact, and $f(x,y)$ a function on $M\times N$ that is quasi-concave-convex and (upper semicontinuous)-(lower semicontinuous). Then 
    \begin{equation*}
        \sup_{\x\in M}\;\inf_{y\in N}f(x,y)= \inf_{y\in N}\;\sup_{\x\in M} f(x,y).
    \end{equation*}
\end{theorem}

The second result was proved in \cite{anari2017nash}.
\begin{lemma}[Lemma 4.3 in \cite{anari2017nash}]
\label{lem:anari_nash_lemma}
Let $p : \RR^{m}_{\geq 0} \rightarrow \R_{\geq 0}$ be a positive function satisfying the following properties:
\begin{itemize}
    \item $p(\alpha \y) = \alpha^n p(\y)$ for all $\y \geq 0$,
    \item $\log p(\y)$ is convex in $\log \y$.
\end{itemize}
Then the following inequality holds
\begin{align*}
    \inf_{\y > 0: y^S \geq 0, \forall S\in \binom{[m]}{n}} \log p(\y) = \sup_{\bm{\alpha}\in [0,1]^{m}, \sum_j \alpha_j =n}\inf_{\y > 0}\quad \log p(\y)-\sum\limits_{j=1}^m\alpha_j \log(y_j).
\end{align*}
\end{lemma}
While the original result in \cite{anari2017nash} assumed $p$ to be a homogeneous polynomial with positive coefficients, their proof only relies on the two properties presented in Lemma \ref{lem:anari_nash_lemma}.

\subsection{Proof of Theorem \ref{thm:sym-eq}} \label{app:sym}
To prove Theorem \ref{thm:sym-eq}, we start with the \eqref{eq:sym-poly} and derive the convex program \eqref{eq:cvx-sym} via a sequence of duals presented in Lemmas \ref{lem:sym-eq-1}, \ref{lem:sym-eq-2}, and \ref{lem:sym-eq-3}.

Let $\mathcal{P}$ and $\mathcal{Q}$ denote the feasible regions for $\x$ and $\y$ in \eqref{eq:sym-poly}, respectively. 
\begin{align*}
    \cP &:=\left\{\mb{x}\in \RR^{\cA\times \cG}_{\geq 0} \,:\,  \sum_{i\in \cA} x_{ij} = 1 \quad \forall j \in \cG  \right\} \\
    \cQ &:=\left\{\mb{y}\in \RR^{\cG}_{>0}\,:\, \prod_{j \in S} y_j \geq 1 \quad \forall S \in \binom{\cG}{n} \right\}.
\end{align*}

Note that the inner function in the objective 
\begin{equation*}
    f(x)= \inf_{\mb{y} \in \cQ}\sum\limits_{i\in \cA} \log\left(\sum_{j \in \cG} x_{ij} \; v_{ij} \;y_j\right),\; 
\end{equation*}
is bounded above ($\mb{y}=\bm{1}$ belongs to $\cQ$), and the domain of $x$, $\cP$, is compact (Bounded and Closed sets in Euclidean space are compact using Heine-Borel Theorem). 

Lemma \ref{lem:LogConcave-Sym-isrelax} shows that the inner infimum of \eqref{eq:sym-poly} is $>-\infty$ for any integral allocation $\mb{x}$ that assigns at least one item to each agent in the support of $\mb{v}$. We know such an allocation exists by Assumption \ref{assumption:1}.
\begin{lemma} 
\label{lem:LogConcave-Sym-isrelax}
    For any integral allocation $\mb{x}\in \cP \cap \{0,1\}^{|\cA|\times |\cG|}$, 
    \begin{align*}
  \inf_{\mb{y}\in \cQ} \quad & \sum_{i \in \cA} \log\left(\sum_{j \in \cG} x_{ij} \, v_{ij} \,y_j\right) =   \sum_{i \in \cA} \log\left(\sum_{j \in \cG} x_{ij} \, v_{ij} \right) .
    \end{align*}
\end{lemma}
\begin{proof}
Let $\sigma: \cG \rightarrow \cA$ be the allocation corresponding to $\mb{x}$, i.e., $\sigma(j) = i$ iff $x_{ij} = 1$ and let $\cS=\{S\in \binom{\cG}{n}: \forall i\in \cA,\, \exists j \in S \textnormal{ such that } x_{ij}=1\}$. For any $\mb{y}\in \cQ$, 
    \begin{align*}
        \sum_{i \in \cA} \log\left(\sum_{j \in \cG} x_{ij} \; v_{ij} \;y_j\right) &=   \log\left( \sum\limits_{S\in \mathcal{S}}y^S\prod_{j \in S} x_{\sigma(j) j}v_{\sigma(j) j}\right)\\
        &\geq  \log\left( \sum\limits_{S\in \mathcal{S}}\prod_{j \in S} x_{\sigma(j)j}v_{\sigma(j)j}\right)=  \sum_{i \in \cA} \log\left(\sum_{j \in \cG} x_{ij} \; v_{ij} \right).
    \end{align*}
    Here, the only inequality holds because  $ y^S \geq 1$ for each $S\in \mathcal{S}$.
    
    Setting $y_j = 1$ for each $j \in \cG$ gives the equality. 
\end{proof}

\begin{lemma}\label{lem:sym-eq-1}
    The optimal value of \eqref{eq:sym-poly} is the same as 
    \begin{equation}
        \inf\limits_{\bm{\delta}} \; \max\limits_{\x\in \mathcal{P}} \; \sum\limits_{i\in \cA}\log\left(\sum\limits_{j\in \cG}x_{ij}\,v_{ij} \,e^{-\delta_j}\right)+\sum\limits_{j\in \cG}\max(0,\,\delta_j) . \label{cp:p1} \tag{Unweighted-Primal} 
    \end{equation}
\end{lemma}
\begin{proof}
For a fixed $\mb{x} \in \mathcal{P}$, using Lemma \ref{lem:anari_nash_lemma} with $p_{\mb{x}}(\y) = \prod_{i\in \cA}\left(\sum\limits_{j\in \cG}x_{ij} \, v_{ij} \, y_j\right)$, we get
\begin{align*}
   \inf_{\y > 0: y^S \geq 0, \forall S\in \binom{\cG}{n}}  \log{p_{\mb{x}}(\mb{y})} &= \inf_{\y > 0: y^S \geq 0, \forall S\in \binom{\cG}{n}}  \sum\limits_{i\in \cA}\log \left(\sum\limits_{j\in \cG}x_{ij} \, v_{ij} \, y_j \right)\\
 &= \sup_{\bm{\alpha}\in [0,1]^{|\cG|}, \sum_j \alpha_j = n} \inf_{\y > 0} \; \sum\limits_{i\in \cA}\log \left(\sum\limits_{j\in \cG}x_{ij}\,v_{ij}\, y_j \right)-\sum\limits_{j\in \cG}\alpha_j \log(y_j) .
\end{align*}
Substituting $\delta_j = -\log(y_j)$ and
taking a maximum over $\x$, we get
\begin{align*}
\max_{\x \in \mathcal{P}}\inf_{\y > 0: y^S \geq 0, \forall S\in \binom{\cG}{n}}  \log{p_{\mb{x}}(\mb{y})}
 &= \sup_{\x \in \mathcal{P},\bm{\alpha}\in [0,1]^{|\cG|}, \sum_j \alpha_j = n}\inf_{\bm{\delta} } \; \sum\limits_{i\in \cA}\log \left(\sum\limits_{j\in \cG}x_{ij} \, v_{ij} \, e^{-\delta_j} \right)+\sum\limits_{j\in \cG}\alpha_j \delta_j .
\end{align*}
As the domains of both $\x$ and $\bm{\alpha}$ are compact, using Theorem \ref{lem:sion-minimax} on the previous equation, we get
\begin{align*}
\max_{\x \in \mathcal{P}} \inf_{\y > 0: y^S \geq 0, \forall S\in \binom{\cG}{n}}  \log{p_{\mb{x}}(\mb{y})}
 &= \inf_{\bm{\delta} } \;\max_{\x \in \mathcal{P}}\;\max_{\bm{\alpha}\in [0,1]^{|\cG|}, \sum_j \alpha_j = n}\; \sum\limits_{i\in \cA}\log \left(\sum\limits_{j\in \cG}x_{ij} \, v_{ij} \, e^{-\delta_j} \right)+\sum\limits_{j\in \cG}\alpha_j \delta_j .
\end{align*}
Finally, the following claim completes the proof.
\begin{align}
   &\inf_{\bm{\delta} } \;\max_{\x \in \mathcal{P}}\;\max_{\bm{\alpha}\in [0,1]^{|\cG|}, \sum_j \alpha_j = n}\; \sum\limits_{i\in \cA}\log \left(\sum\limits_{j\in \cG}x_{ij} \, v_{ij} \, e^{-\delta_j} \right)+\sum\limits_{j\in \cG}\alpha_j \delta_j \notag \\
   &= \inf\limits_{\bm{\delta}} \; \max\limits_{\x\in \mathcal{P}} \quad\sum\limits_{i\in \cA}\log\left(\sum\limits_{j\in \cG}x_{ij}v_{ij}e^{-\delta_j}\right)+\sum\limits_{j\in \cG}\max(0,\delta_j). \label{eq:fin-eq}
\end{align}
For proving the claim, we define functions \begin{align*}
    f_1(\bm{\delta}, \x, \bm{\alpha}) &= \sum\limits_{i\in \cA}\log \left(\sum\limits_{j\in \cG}x_{ij} \, v_{ij} \, e^{-\delta_j} \right)+\sum\limits_{j\in \cG}\alpha_j \delta_j, \quad \text{and} \\
    f_2(\bm{\delta}, \x) &= \sum\limits_{i\in \cA}\log\left(\sum\limits_{j\in \cG}x_{ij}v_{ij}e^{-\delta_j}\right)+\sum\limits_{j\in \cG}\max(0,\delta_j).
\end{align*}

Observe that for any $\bm{\delta}$ and $\bm{\alpha} \in [0,1]^{|G|}$, $\alpha_j \delta_j \leq \max(0, \delta_j)$. Therefore, for any $\bm{\delta}, \x$ and $\bm{\alpha} \in [0,1]^{|G|}$, we have $f_1(\bm{\delta}, \x, \bm{\alpha}) \leq f_2(\bm{\delta}, \x)$. As a result, 
\begin{equation}
     \inf_{\bm{\delta} } \;\max_{\x \in \mathcal{P}}\;\max_{\bm{\alpha}\in [0,1]^{|\cG|}, \sum_j \alpha_j = n}\; f_1(\bm{\delta}, \x, \bm{\alpha}) \leq \inf\limits_{\bm{\delta}} \; \max\limits_{\x\in \mathcal{P}} \; f_2(\bm{\delta}, \x). \label{eq:eq:fin1}
\end{equation}

To establish an inequality in the other direction, first note that $f_1(\bm{\delta}, \x, \bm{\alpha}) = f_1(\bm{\delta} + t\cdot \bm{1}, \x, \bm{\alpha})$ for any $t \in \RR$. 
So, for a fixed $\bm{\delta}$, let $t_{\delta}$ denote a value of $t$ for which  the $n$ largest values of $\bm{\delta}+ t_{\delta} \cdot \bm{1}$ are non-negative and the $m-n$ smallest values of $\delta$ are non-positive. Then 
\begin{align}
  \max_{\bm{\alpha}\in [0,1]^{|\cG|}, \sum_j \alpha_j = n}\;   f_1(\bm{\delta}, \x, \bm{\alpha}) &= \max_{\bm{\alpha}\in [0,1]^{|\cG|}, \sum_j \alpha_j = n}\;   f_1(\bm{\delta} + t_{\delta} \cdot \bm{1}, \x, \bm{\alpha}) \notag\\
  &= \sum_{i \in \cA} \log\left(\sum_{j \in G} x_{ij} v_{ij} e^{-\delta_j - t_{\delta} }\right) +  \max_{\bm{\alpha}\in [0,1]^{|\cG|}, \sum_j \alpha_j = n}\sum_{j \in \cG} \alpha_j (\delta_j+t_{\delta}).\label{eq:fin3}
\end{align}
The term $\sum_{j \in \cG} \alpha_j (\delta_j+t_{\delta})$ is maximized when $\alpha_j = 1$ for the largest $n$ coordinates of $\bm{\delta} + t \cdot \bm{1}$. As a result, we get 
\begin{align}
     \sum_{i \in \cA} \log\left(\sum_{j \in G} x_{ij} v_{ij} e^{-\delta_j - t_{\delta} }\right) +  \sum_{j \in \cG} \max(0, \delta_j + t_{\delta}) = f_2(\bm{\delta} + t_{\delta}\cdot \bm{1}, \x) .\label{eq:fin4}
\end{align}
Combining equations \eqref{eq:fin3} and \eqref{eq:fin4}, and taking max over $\x$, we have
\begin{align*}
  \max_{\x \in \mathcal{P}} \; \max_{\bm{\alpha}\in [0,1]^{|\cG|}, \sum_j \alpha_j = n}\;   f_1(\bm{\delta}, \x, \bm{\alpha})  &= \max_{\x \in \mathcal{P}}\; f_2(\bm{\delta} + t_{\delta}\cdot \bm{1}, \x) \geq \inf_{\bm{\gamma}}\;\max_{\x \in \mathcal{P}}\; f_2(\bm{\gamma}, \x).
\end{align*}

Taking an infimum over $\bm{\delta}$, we obtain 
\begin{align}
 \inf_{\bm{\delta}}\;\max_{\x \in \mathcal{P}}\;\max_{\bm{\alpha}\in [0,1]^{|\cG|}, \sum_j \alpha_j = n}\;   f_1(\bm{\delta} , \x, \bm{\alpha}) &\geq  \inf_{\bm{\delta}}\; \inf_{\bm{\gamma}}\;\max_{\x \in \mathcal{P}} \;f_2(\bm{\gamma}, \x) \notag \\
  &= \inf_{\bm{\gamma}}\;\max_{\x \in \mathcal{P}} \;f_2(\bm{\gamma}, \x) . \label{eq:fin2}
\end{align}
Here, the last equality follows as the function being optimized does not depend on $\bm{\delta}$.

Combining equations \eqref{eq:eq:fin1} and \eqref{eq:fin2} completes the proof of equation \eqref{eq:fin-eq}.
\end{proof}

\begin{lemma} \label{lem:sym-eq-2}
\label{lem: EG-Sym-Dual-Shmyrev-Sym}
   The optimal values of \eqref{cp:p1} is the same as that of the following program.
  \begin{align}
\label{cp:fSR-dual}
   \inf_{\bm{\delta}, \mb{r}, \bm{\gamma}} \quad &\sum\limits_{j\in \cG}e^{r_{j}} +\sum\limits_{i\in \cA}\gamma_{i} +\sum\limits_{j\in \cG}\delta_{j}-n \tag{Unweighted-Dual}\\ 
    \quad &r_j+\gamma_i+\delta_j \geq \log v_{ij}\quad \forall (i,j) \in \cA\times \cG  \nonumber \\
    \quad &\bm{\delta}\geq \bm{0}. \nonumber
\end{align}
 \end{lemma}
\begin{proof}
    For a fixed $\bm{\delta}$, let us first re-write the internal maximum of \eqref{cp:p1} as 
    \begin{align}
\label{cp:EG-Symmetric}
    \max_{\x, \mb{u}} \quad & \sum\limits_{i\in \cA} \log u_i + f(\bm{\delta}) \\
    \quad &u_i \leq \sum\limits_{j\in \cG}x_{ij}\, v_{ij} \,e^{-\delta_j} \quad \forall i \in \mathcal{\cA} \nonumber\\
    \quad &\sum\limits_{i\in \cA}x_{ij} \leq 1 \quad \forall j \in \cG \nonumber \\ 
    \quad & \mb{x} \geq \bm{0}\,, \nonumber
\end{align}
where $f(\bm{\delta}) = \sum_{j \in \cG} \max(0, \delta_j)$. 

Let $\beta_i$, $p_j$, and $\theta_{ij}$ be the Lagrange dual variables associated with the constraints corresponding to agent $i$, item $j$, and agent-item pair$(i,j)$, respectively. The Lagrangian of the above convex program is defined as follows
\begin{align*}
\label{cp:EG-sym}
L(\x, \mb{u}, \bm{\beta}, \bm{\theta}, \mb{p}) &= f(\bm{\delta}) +  \left[\sum\limits_{i\in \cA} \log u_i +\sum\limits_{i\in \cA}\beta_i \left(\sum\limits_{j\in \cG}x_{ij}\, v_{ij} \,e^{-\delta_j} -u_i\right)+\sum\limits_{j\in \cG}p_j(1-\sum\limits_{i\in \cA}x_{ij})+\sum\limits_{i,j}\theta_{ij}x_{ij}\right] \\
&= f(\bm{\delta})+\left[\sum\limits_{i\in \cA} (\log u_i-\beta_i u_i) +\sum_{i \in \cA} \sum_{j \in \cG} x_{ij}\left(\beta_{i}\, v_{ij} \,e^{-\delta_j}+\theta_{ij}-p_j\right)+\sum\limits_{j\in \cG}p_j\right]  .
\end{align*}
The Lagrange dual of \eqref{cp:EG-Symmetric} is given by 
\begin{equation}
   g(\bm{\beta},\bm{\theta},\bm{p}) =  \max_{\x \in \mathcal{P},\mb{u} \geq 0} L(\x, \mb{u}, \bm{\beta}, \bm{\theta}, \mb{p}). \label{eq:dual1}
\end{equation}
Observe that solution $x_{ij}=1/n$ for each $(i,j) \in \cA \times \cG$ lies in the relative interior of $\cP$. Since all the constraints are affine, Slater's condition is satisfied for \eqref{cp:p1}. Thus, the optimal value of the infimum of Lagrange dual over $\bm{\beta},\bm{\theta},\bm{p} \geq 0$ is exactly equal to the optimum of \eqref{cp:EG-Symmetric}.

The KKT conditions imply that the optimal solutions must satisfy
\begin{align*}
    \frac{1}{u_i} - \beta_i &= 0 \quad \forall i \in \cA \\
    \beta_i \, v_{ij} \,e^{-\delta_j} - \sum_{j \in \cG} p_j + \theta_{ij} &= 0 \quad \forall (i,j) \in \cA \times \cG.
\end{align*} 
The KKT conditions imply that $u_i = 1/\beta_i$ for each $i\in \cA$ maximizes the Lagrangian. For the supremum over $\mb{x}, \mb{u}$ in \eqref{eq:dual1} to stay finite, the second KKT condition is necessary and sufficient. Substituting these conditions in the Langrangian gives the following convex program.
\begin{align*}
\inf_{\mb{p}, \bm{\beta}, \bm{\theta}} &\quad f(\bm{\delta}) + \sum\limits_{j\in \cG}p_j -\sum\limits_{i\in \cA}\log \beta_i -n \\
\quad & p_j = \beta_{i} \, v_{ij} \, e^{-\delta_j} + \theta_{ij} \quad \forall (i,j) \in \cA \times \cG \\
\quad &\mb{p}, \bm{\beta}, \bm{\theta} \geq \mb{0}.
\end{align*}

Observe that we can remove $\bm{\theta}$ from the above program while making the first constraint an inequality.
By substituting $r_{j}=\log p_j, \gamma_i =-\log \beta_i $, the above program is equivalent to 
\begin{align*}
    \inf_{\mb{r}, \bm{\gamma}}\quad & f(\bm{\delta}) + \sum\limits_{j\in \cG} e^{r_j}+\sum\limits_{i\in \cA}\gamma_i +\sum\limits_{j\in \cG} -n \\
    \quad &r_{j}+\gamma_i+\delta_j \geq \log v_{ij} \quad \forall (i,j) \in \cA\times \cG.
\end{align*}
As \eqref{cp:p1} involves an infimum over $\bm{\delta}$, whenever $\delta_j<0$, we can increase it to $\delta_j=0$ without increasing the value of $f(\bm{\delta})$ and maintaining feasibility. Using this observation and taking an infimum over $\bm{\delta}$, the above program gives \eqref{cp:fSR-dual}.
\end{proof}

 \begin{lemma} \label{lem:sym-eq-3}
     The optimal values of \eqref{cp:fSR-dual} and \eqref{eq:cvx-sym} are the same.
 \end{lemma}
\begin{proof}
Let $b_{ij}$ be the Lagrange dual variable associated with constraint $r_j+\gamma_i+\delta_j \geq \log v_{ij}$ of \eqref{cp:fSR-dual} and let $\tau_{ij}$ be the Lagrange dual variable associated with constraint $\delta_{ij} \geq 0$. The Lagrangian of \eqref{cp:fSR-dual} is defined as follows
\begin{align*}
   L(\mb{r},\bm{\gamma},\bm{\delta},\mb{b},\bm{\tau}) &= \sum\limits_{j\in \cG}e^{r_j}+\sum\limits_{i\in \cA}\gamma_{i} +\sum\limits_{j\in \cG}\delta_{j}-n+\sum\limits_{i,j}b_{ij}(\log v_{ij}-r_j-\gamma_i-\delta_j)-\sum\limits_{j\in \cG}\delta_j\tau_j  \\
    &=\sum\limits_{j\in \cG}(e^{r_j}-(\sum\limits_{i\in \cA}b_{ij})r_j) +\sum\limits_{i\in \cA}\gamma_i(1-\sum\limits_{j\in \cG}b_{ij}) \sum\limits_{j\in \cG}\delta_j(1-\tau_j-\sum\limits_{i\in \cA}b_{ij})+\sum\limits_{i,j}b_{ij}\log v_{ij} -n .
\end{align*}
The Lagrange dual of \eqref{cp:fSR-dual} is given by 
\begin{equation}
    g(\mb{b},\bm{\tau}) = \inf\limits_{\bm{\delta} \geq 0,\mb{r},\bm{\gamma}} L(\mb{r},\bm{\gamma},\bm{\delta},\mb{b},\bm{\tau}). \label{eq:dual2}
\end{equation}
One can verify that Slater's condition is satisfied by \eqref{cp:fSR-dual}. So, the supremum of \eqref{eq:dual2} with $\mb{b},\bm{\tau} \geq 0$ is equal to the optimum of \eqref{cp:fSR-dual}.

The KKT conditions for the Langrangian give
\begin{align*}
    e^{r_j} - \sum_{i\in \cA} b_{ij} = 0 \quad \quad 1-\sum_{j \in \cG} b_{ij} = 0 \quad \quad
    1-\tau_j - \sum_{j\in A} b_{ij} = 0.
\end{align*}
The KKT conditions imply $r_{j}= \log\left(\sum\limits_{i\in \cA}b_{ij}\right)$ for each $j \in \cG$ minimizes the Lagrangian. For the infimum over $\bm{\gamma}, \bm{\delta}$ in \eqref{eq:dual2} to stay finite, the conditions $1=\sum\limits_{j\in \cG}b_{ij}$ and $1-\tau_{j}=\sum\limits_{i\in \cA}b_{ij}$ are necessary and sufficient. Substituting these conditions in the Lagrangian, we get 
\begin{align*}
    \sup_{\mb{b}, \bm{\tau}}\quad & \sum\limits_{i,j}b_{ij}\log v_{ij}-\sum\limits_{j\in \cG} \sum\limits_{i\in \cA}b_{ij} \log \left(\sum\limits_{i\in \cA}b_{ij}\right)+ \sum\limits_{j\in \cG} \sum\limits_{i\in \cA}b_{ij}-n \\
    \quad &\sum\limits_{j\in \cG}b_{ij}=1 \\
    \quad &\sum\limits_{i\in \cA}b_{ij}=1-\tau_{j} \\
    \quad &\mb{b}, \bm{\tau} \geq \mb{0} .
\end{align*}
Observe that the supremum in the above program can be switched to maximum as the feasible region is compact and the objective is bounded. Also note that $\sum_{i,j} b_{ij} = n$ for any $\mb{b}$ in the feasible region. As a result, the last two terms in the objective cancel each other. Finally, on substituting $q_j = \sum_{i\in \cA} b_{ij}$ in the above program, we obtain \eqref{eq:cvx-sym}.
\end{proof}

\subsection{Generalization to Weighted Nash Social Welfare} \label{app:asym}
\label{sec:asymmetric}
Given an instance of weighted Nash Social Welfare $(\cA, \cG, \mb{v}, \mb{w})$ where $\sum\limits_{i\in \cA}w_i =1$ and $\mb{w}\geq \mb{0}$, we introduce the following program as a generalization of \eqref{eq:sym-poly} program.
\begin{align}
\label{cp:LogConcave-Asym}
    \max_{\mb{x} \geq 0}\min_{\mb{y} > 0} \quad &  \;\sum_{i \in \cA}w_i \log\left(\sum_{j \in \cG} x_{ij} \; v_{ij} \;y_j^{1/w_{i}}\right) \tag{LogConcave-Weighted}\\
    \mathrm{s.t.} \quad &  \sum_{i\in \cA} x_{ij} = 1 \quad \forall j \in \cG\notag\\
      \quad & \prod_{j \in S} y_j \geq 1 \quad \forall S \in \binom{\cG}{n}. \notag
\end{align}
Observe that the feasible region of \eqref{cp:LogConcave-Asym} is given by $\x \in \mathcal{P}$ and $\y \in \mathcal{Q}$, which is identical to that of \eqref{eq:sym-poly}.

The main result of this section is the following.
\begin{theorem}\label{thm:asym-eq}
    The optimal values of \eqref{cp:LogConcave-Asym} and \nameref{eq:b-cvx} are the same.
\end{theorem}

We prove Theorem \ref{thm:asym-eq} analogously to Theorem \ref{thm:sym-eq}, starting with \eqref{cp:LogConcave-Asym} and deriving \nameref{eq:b-cvx} via a sequence of duals presented in Lemmas \ref{lem:asym-eq-1}, \ref{lem:asym-eq-2}, and \ref{lem:asym-eq-3}.

We start by establishing that \ref{cp:LogConcave-Asym} is indeed a relaxation of the weighted Nash Social Welfare, and the inner infimum is bounded in the following lemma.
\begin{lemma}
\label{lem:LogConcave-ASym-isrelax}
    For any integral allocation $\mb{x}\in \cP \cap \{0,1\}^{|\cA|\times |\cG|}$, 
    \begin{align*}
  \inf_{\mb{y}\in \cQ} \quad & \sum_{i \in \cA} w_i\log\left(\sum_{j \in \cG} x_{ij} \, v_{ij} \,y_j^{1/w_i}\right) =   \sum_{i \in \cA}w_i \log\left(\sum_{j \in \cG} x_{ij} \, v_{ij} \right)
    \end{align*}
\end{lemma}
\begin{proof}
For each $i$, let $S_i = \{ j \in \cG: x_{ij} = 1\}$ be the allocation corresponding to $\mb{x}$.
Then for any $\mb{y}\in \cQ$, 
    \begin{align}
        \sum_{i \in \cA} w_i\log\left(\sum_{j \in \cG} x_{ij} \, v_{ij} \,y_j^{1/w_i}\right) &-  \sum_{i \in \cA}w_i \log\left(\sum_{j \in \cG}  x_{ij} \, v_{ij}\right)=  \sum_{i \in \cA}w_i \log\left(\frac{\sum_{j \in S_i}  v_{ij} \;y_j^{1/w_i}}{\sum\limits_{j\in S_i}v_{ij}}\right). \label{eq:w-difference}
    \end{align}

Now for positive reals $c_1,\dots,c_m$ with $\sum\limits_{j=1}^m c_j =1$, and $0 \leq p \leq q$, the weighted power mean inequality states that for any $\mb{z}\in \RR_{\geq 0}^m$,
   \begin{align}
       \left(\sum\limits_{j=1}^m c_jz_j^p\right)^{1/p} \leq \left(\sum\limits_{j=1}^m c_jz_j^q\right)^{1/q} . \label{eq:wpm}
   \end{align}
This inequality follows from Jensen's inequality.

For each $i\in \cA$, define $q_i = \frac{1}{w_i}$ and $c^{(i)}_j = \frac{v_{ij}}{\sum\limits_{j\in S_i}v_{ij}}$ for every $j \in S_i$. Since $q_i = \frac{1}{w_i} \geq 1$, using equation \ref{eq:wpm}, we get 
    \begin{align*}
        w_i \log\left(\frac{\sum_{j \in S_i}  v_{ij} \;y_j^{1/w_i}}{\sum\limits_{j\in S_i}v_{ij}}\right) \geq \log\left(\frac{\sum_{j \in S_i}  v_{ij} \;y_j}{\sum\limits_{j\in S_i}v_{ij}} \right)
       \geq 0
    \end{align*}
    for each agent $i$. Summing this inequality over all agents and substituting in \eqref{eq:w-difference} gives
    \begin{align*}
         \sum_{i \in \cA} w_i\log\left(\sum_{j \in \cG} x_{ij} \; v_{ij} \;y_j^{1/w_i}\right) \geq \sum_{i \in \cA}w_i \log\left(\sum_{j \in \cG}  x_{ij}v_{ij}\right).
    \end{align*}
    Observe that equality holds when $y_j=1$ for all $j \in \cG$.
\end{proof}

\begin{lemma} \label{lem:asym-eq-1}
    The optimal value of \eqref{cp:LogConcave-Asym} is the same as 
    \begin{align}\label{cp:Weighted_primal_1} 
        \inf\limits_{\bm{\delta}} \max\limits_{\x\in \mathcal{P}} \; \sum\limits_{i\in \cA}w_i\log\left(\sum\limits_{j\in \cG}x_{ij}v_{ij}e^{-\delta_j/w_i}\right)+\sum\limits_{j\in \cG}\max(0,\delta_j).  \tag{Weighted-Primal}
    \end{align}
\end{lemma}
The following fact is crucial to the proof of this lemma.
\begin{fact} \label{claim:log-convex}
    Let $p(\y) = w \log\left(\sum_{j=1}^m c_j \,y_j^{1/w}\right)$ with $w > 0$ and $c_j \geq 0$ for each $j$. Then $\log p(\y)$ is a convex function in $\log(\y)$.
\end{fact}
\begin{proof}
For a fixed $\mb{x} \in \mathcal{P}$, the function 
\begin{equation*}
    p_{\x}(\y) = \prod_{i\in \cA}\left(\sum\limits_{j\in \cG}x_{ij} \, v_{ij} \, y_j^{1/w_i}\right)^{w_i}
\end{equation*} satisfies all the prerequisites of Lemma \ref{lem:anari_nash_lemma}. The first property is easy to verify and the second property follows from Fact \ref{claim:log-convex}. Therefore, by Lemma \ref{lem:anari_nash_lemma}, we get
\begin{align*}
   \inf_{\y > 0: y^S \geq 0, \forall S\in \binom{\cG}{n}} \log{p_{\x}(\y)} &= \inf_{\y > 0: y^S \geq 0, \forall S\in \binom{\cG}{n}}  \sum\limits_{i\in \cA}w_i \log \left(\sum\limits_{j\in \cG}x_{ij} \, v_{ij} \, y_j^{1/w_i} \right)\\
 &= \sup_{\bm{\alpha}\in [0,1]^{|\cG|}, \sum_j \alpha_j = n} \inf_{\y > 0} \; \sum\limits_{i\in \cA}w_i\log \left(\sum\limits_{j\in \cG}x_{ij}\,v_{ij}\, y_j^{1/w_i} \right)-\sum\limits_{j\in \cG}\alpha_j \log(y_j) .
\end{align*}
Substituting $\delta_j = -\log(y_j)$, and taking the supremum over $\x$, we get
\begin{align*}
\max_{\x \in \mathcal{P}} \inf_{\y > 0: y^S \geq 0, \forall S\in \binom{\cG}{n}} \log{p_{\x}(\y)} &= \sup_{\x \in \mathcal{P},\bm{\alpha}\in [0,1]^{|\cG|}, \sum_j \alpha_j = n}\inf_{\bm{\delta} } \; \sum\limits_{i\in \cA} w_i \log \left(\sum\limits_{j\in \cG}x_{ij} \, v_{ij} \, e^{-\delta_j/w_i} \right)+\sum\limits_{j\in \cG}\alpha_j \delta_j .
\end{align*}
As the domains of both $\x$ and $\bm{\alpha}$ are compact, using Theorem \ref{lem:sion-minimax}, we get
\begin{align*}
\max_{\x \in \mathcal{P}} \inf_{\y > 0: y^S \geq 0, \forall S\in \binom{\cG}{n}} &\sum\limits_{i\in \cA} w_i\log \left(\sum\limits_{j\in \cG}x_{ij} \, v_{ij} \, y_j^{1/w_i} \right)\\
 &= \inf_{\bm{\delta} } \;\max_{\x \in \mathcal{P}} \;\max_{\bm{\alpha}\in [0,1]^{|\cG|}, \sum_j \alpha_j = n}\; \sum\limits_{i\in \cA}w_i\log \left(\sum\limits_{j\in \cG}x_{ij} \, v_{ij} \, e^{-\delta_j/w_i} \right)+\sum\limits_{j\in \cG}\alpha_j \delta_j .
\end{align*}

Finally, we claim that 
\begin{align*}
   &\inf\limits_{\bm{\delta}}\; \max\limits_{\x\in \mathcal{P}}\;\max_{\bm{\alpha} \in [0,1]^{|\cG|}, \sum_j \alpha_j = n} \quad \sum\limits_{i\in \cA}w_i\log\left(\sum\limits_{j\in \cG}x_{ij}\, v_{ij}\,e^{-\delta_j/w_i}\right)+\sum\limits_{j\in \cG}\alpha_{j}\delta_j  \\
  &= \inf\limits_{\bm{\delta}} \; \max\limits_{\x\in \mathcal{P}} \quad\sum\limits_{i\in \cA} w_i \log\left(\sum\limits_{j\in \cG}x_{ij}\,v_{ij}\,e^{-\delta_j/w_i}\right)+\sum\limits_{j\in \cG}\max(0,\delta_j). 
\end{align*}
The proof of this claim is identical to the proof of the unweighted case in equation \eqref{eq:fin-eq}.
\end{proof}

\begin{lemma} \label{lem:asym-eq-2}
    The optimal value of \eqref{cp:Weighted_primal_1} is the same as that of the following program.
    \begin{align} \label{cp:weighted_dual_1}
        \inf_{\bm{\delta}, \mb{r}, \bm{\gamma}} \quad &\sum\limits_{j\in \cG} e^{r_j}+\sum\limits_{i\in \cA}w_i\gamma_i +\sum\limits_{j\in \cG}  \delta_j+\sum\limits_{i\in \cA}(w_i\log w_i-w_i) \tag{Weighted-Dual} \\
    \quad &r_{j}+\gamma_i+\frac{\delta_j}{w_i} \geq \log v_{ij} \quad \forall (i, j) \in \cA\times \cG. \notag
    \end{align} 
\end{lemma}
\begin{proof}
    For a fixed $\bm{\delta}$, let us first re-write the internal maximum of \eqref{cp:Weighted_primal_1} as 
    \begin{align}
\label{cp:EG-weighted}
    \max_{\x, \mb{u}} \quad & \sum\limits_{i\in \cA} w_i \log u_i + f(\bm{\delta}) \\
    \quad &u_i \leq \sum\limits_{j\in \cG}x_{ij}\, v_{ij} \,e^{-\delta_j/w_i} \quad \forall i \in \mathcal{\cA} \nonumber\\
    \quad &\sum\limits_{i\in \cA}x_{ij} \leq 1 \quad \forall j \in \cG \nonumber \\ 
    \quad & \mb{x} \geq \bm{0}\,, \nonumber
\end{align}
where $f(\bm{\delta}) = \sum_{j \in \cG} \max(0, \delta_j)$. 

Let $\beta_i$, $p_j$, and $\theta_{ij}$ be the Lagrange dual variables associated with the constraints corresponding to agent $i$, item $j$, and agent-item pair$(i,j)$, respectively. The Lagrangian of the above convex program is defined as follows
\begin{align*}
L(\x, &\mb{u}, \bm{\beta}, \bm{\theta}, \mb{p}) \\
&= f(\bm{\delta}) +  \sum\limits_{i\in \cA} w_i\log u_i +\sum\limits_{i\in \cA}\beta_i \left(\sum\limits_{j\in \cG}x_{ij}\, v_{ij} \,e^{-\delta_j/w_i} -u_i\right) +\sum\limits_{j\in \cG}p_j\left(1-\sum\limits_{i\in \cA}x_{ij}\right)+\sum\limits_{i,j}\theta_{ij}x_{ij} \\
&= f(\bm{\delta})+\left[\sum\limits_{i\in \cA} (w_i\log u_i-\beta_i u_i) +\sum\limits_{i,j}x_{ij}\left(\beta_{i}\, v_{ij} \,e^{-\delta_j/w_i}+\theta_{ij}-p_j\right)+\sum\limits_{j\in \cG}p_j\right].
\end{align*}
The Lagrange dual of \eqref{cp:EG-weighted} is given by 
\begin{equation*}
   g(\bm{\beta},\bm{\theta},\bm{p}) =  \max_{\x \in \mathcal{P},\mb{u} \geq 0} L(\x, \mb{u}, \bm{\beta}, \bm{\theta}, \mb{p}).
\end{equation*}
Observe that solution $x_{ij}=1/n$ is in the relative interior of $\cP$. Since all the constraints are affine, Slater's condition is satisfied. Thus the optimum value of the infimum of Lagrange dual over $\bm{\beta},\bm{\theta},\bm{p} \geq 0$ is exactly equal to the optimum of \eqref{cp:EG-weighted}.

The KKT conditions for the Lagrangian imply
\begin{align*}
    \frac{w_i}{u_i} - \beta_i &= 0 \quad \forall i \in \cA \\
    \beta_i \, v_{ij} \,e^{-\delta_j/w_i} - \sum_{j \in \cG} p_j + \theta_{ij} &= 0 \quad \forall (i,j) \in \cA \times \cG.
\end{align*} 
The KKT conditions imply that $u_i = w_i/\beta_i$ for each $i\in \cA$ maximizes the Lagrangian. For the supremum over $\mb{x}, \mb{u}$ in \eqref{eq:dual2} to stay finite, the second KKT condition is necessary and sufficient. Substituting these conditions in the Langrangian gives the following convex program.
\begin{align*}
\inf_{\mb{p}, \bm{\beta}, \bm{\theta}} &\quad f(\bm{\delta}) + \sum\limits_{j\in \cG}p_j +\sum_{i \in \cA} (w_i \log w_i-w_i)-\sum\limits_{i\in \cA}w_i\log \beta_i  \\
\quad & p_j = \beta_{i} \, v_{ij} \, e^{-\delta_j/w_i} + \theta_{ij} \quad \forall (i,j) \in \cA \times \cG \\
\quad &\mb{p}, \bm{\beta}, \bm{\theta} \geq 0
\end{align*}

Observe that we can remove $\bm{\theta}$ from the above program while making the first constraint an inequality.
By substituting $r_{j}=\log p_j, \gamma_i =-\log \beta_i $, the above program is equivalently to
\begin{align*}
    \inf_{\mb{r}, \bm{\gamma}}\quad & f(\bm{\delta}) + \sum\limits_{j\in \cG} e^{r_j}+\sum\limits_{i\in \cA}w_i \,\gamma_i +\sum\limits_{i\in \cA} (w_i \log{w_i} - w_i) \\
    \quad &r_{j}+\gamma_i+\frac{\delta_j}{w_i} \geq \log v_{ij} \quad \forall (i,j) \in \cA\times \cG.
\end{align*}
As \eqref{cp:Weighted_primal_1} involves an infimum over $\bm{\delta}$, whenever $\delta_j<0$, we can increase it to $\delta_j=0$ without increasing the value of $f(\bm{\delta})$ and maintaining feasibility in the above program. Using this observation and taking an infimum over $\bm{\delta}$ gives \eqref{cp:weighted_dual_1}.
\end{proof}
\begin{lemma} \label{lem:asym-eq-3}
    The optimal value of \eqref{cp:weighted_dual_1} is the same as that of  \nameref{eq:b-cvx}.
\end{lemma}
\begin{proof}
Let $\hat{b}_{ij}$ be the Lagrange dual variable associated with constraint $r_j+\gamma_i+\delta_j \geq \log v_{ij}$ of \eqref{cp:fSR-dual} and let $\mb{y}_{ij}$ be the Lagrange dual variable associated with constraint $\delta_{ij} \geq 0$. The Lagrangian of \eqref{cp:weighted_dual_1} is defined as follows
\begin{align*}
   L(\mb{r},\bm{\gamma},\bm{\delta},\mb{\hat{b}},\bm{\tau}) &= \sum\limits_{j\in \cG}e^{r_j}+\sum\limits_{i\in \cA}w_i \, \gamma_{i} +\sum\limits_{j\in \cG}\delta_{j}+\sum\limits_{i,j}\hat{b}_{ij}(\log v_{ij}-r_j-\gamma_i-\frac{\delta_j}{w_i}) \\
   &\quad-\sum\limits_{j\in \cG}\delta_j\tau_j + \sum_{i\in \cA} (w_i\log{w_i} - w_i ) \\
    &=\sum\limits_{j\in \cG}(e^{r_j}-(\sum\limits_{i\in \cA}\hat{b}_{ij})r_j) +\sum\limits_{i\in \cA}\gamma_i(w_i-\sum\limits_{j\in \cG}\hat{b}_{ij}) +\sum\limits_{j\in \cG}\delta_j(1-\tau_j-\sum\limits_{i\in \cA}\frac{\hat{b}_{ij}}{w_i})\\
    &\quad+\sum\limits_{i,j}\hat{b}_{ij}\log v_{ij} +\sum_{i\in \cA} (w_i\log{w_i} - w_i )  .
\end{align*}
The Lagrange dual of \eqref{cp:weighted_dual_1} is given by 
\begin{equation}
    g(\mb{\hat{b}},\bm{\tau}) = \inf\limits_{\bm{\delta} \geq 0,\mb{r},\bm{\gamma}} L(\mb{r},\bm{\gamma},\bm{\delta},\mb{\hat{b}},\bm{\tau}). \label{eq:dual4}
\end{equation}
One can verify that Slater's condition is satisfied by \eqref{cp:weighted_dual_1}. So, the supremum of \eqref{eq:dual4} with $\mb{b},\bm{\tau} \geq 0$ is equal to the optimum of \eqref{cp:weighted_dual_1}.

The KKT conditions for the Langrangian imply
\begin{align*}
    e^{r_j} - \sum_{i\in \cA} \hat{b}_{ij} &= 0\\
    w_i-\sum_{j \in \cG} \hat{b}_{ij} &= 0 \\
    1-\tau_j - \sum_{j\in A} \frac{\hat{b}_{ij}}{w_i} & = 0.
\end{align*}
The KKT conditions imply that the minimizer for $r_j$ is given by $r_{j}= \log\left(\sum\limits_{i\in \cA}\hat{b}_{ij}\right)$. For the infimum over $\bm{\gamma}, \bm{\delta}$ to stay finite, the conditions $w_i=\sum\limits_{j\in \cG}\hat{b}_{ij}$ for each $i \in \cA$ and $1-\tau_{j}=\sum\limits_{i\in \cA}\hat{b}_{ij}$ for each $j\in \cG$ are necessary and sufficient. Substituting these conditions in the Lagrange dual, we get 
\begin{align*}
    \sup_{\mb{\hat{b}}, \bm{\tau}}\quad & \sum\limits_{i,j}\hat{b}_{ij}\log v_{ij}-\sum\limits_{j\in \cG}\sum\limits_{i\in \cA}\hat{b}_{ij}\log\left(\sum\limits_{i\in \cA}\hat{b}_{ij}\right)+ \sum\limits_{j\in \cG}\sum\limits_{i\in \cA}\hat{b}_{ij}+\sum_{i\in \cA} (w_i \log{w_i} - w_i) \\
    \quad &\sum\limits_{j\in \cG}\hat{b}_{ij}=w_i \\
    \quad &\sum\limits_{i\in \cA} \frac{\hat{b}_{ij}}{w_i}=1-\tau_{j} \\
    \quad &\mb{\hat{b}}, \bm{\tau} \geq \mb{0} .
\end{align*}
Observe that the supremum in the above program can be switched to maximum because the feasible region is compact and the objective is bounded.  
Also, $\sum_{i,j} \hat{b}_{ij} = \sum_{i\in \cA} w_i$ for any feasible $\mb{\hat{b}}$.  
Finally, substituting $b_{ij} = \frac{\hat{b}_{ij}}{w_i}$ and $q_j = \sum_{i \in \cA} b_{ij}$, we obtain  \nameref{eq:b-cvx}.
\end{proof}

\end{document}